\documentclass[a4paper,onecolumn,11pt,accepted=2024-12-05]{quantumarticle}
\pdfoutput=1
\usepackage[utf8]{inputenc}
\usepackage[english]{babel}
\usepackage[T1]{fontenc}
\usepackage{amsmath}
\usepackage{amssymb}
\usepackage{hyperref}
\usepackage{caption,subcaption}
\usepackage[numbers,sort&compress]{natbib}
\usepackage{tikz}
\usepackage{lipsum}

\renewcommand{\i}{{\rm i}}
\newcommand{\bbbone}{\mathchoice {\rm 1\mskip-4mu l} {\rm 1\mskip-4mu l}
{\rm 1\mskip-4.5mu l} {\rm 1\mskip-5mu l}}
\renewcommand{\d}{{\rm d}}
\newcommand{\rx}{{\mathbb R}}
\newcommand{\cx}{{\mathbb C}}

\newcommand{\tr}{{\rm Tr}}
\renewcommand{\r}{{\rm R}}
\newcommand{\s}{{\rm S}}
\newcommand{\h}{{\mathfrak h}}
\renewcommand{\H}{{\mathcal H}}

\newtheorem{prop}{Proposition}
\newtheorem{lem}{Lemma}
\newtheorem{thm}{Theorem}
\newtheorem{cor}{Corollary}

\begin{document}

\title{Quasi-classical limit of a spin coupled to a reservoir}

\author{Michele Correggi}
\email{michele.correggi@gmail.com}
\homepage{https://sites.google.com/view/michele-correggi}
\affiliation{Dipartimento di Matematica, Politecnico di Milano, P.zza Leonardo da Vinci, 32, 20133 Milano, Italy}

\author{Marco Falconi}
\email{marco.falconi@polimi.it}
\homepage{https://www.mfmat.org/}
\affiliation{Dipartimento di Matematica, Politecnico di Milano, P.zza Leonardo da Vinci, 32, 20133 Milano, Italy}

\author{Michele Fantechi}
\email{michele.fantechi@gmail.com}
\affiliation{Dipartimento di Matematica, Politecnico di Milano, P.zza Leonardo da Vinci, 32, 20133 Milano, Italy}

\author{Marco Merkli}
\affiliation{Department of Mathematics and Statistics, Memorial University of Newfoundland, NL A1C 5S7, St.~John's, Canada}  
\email{merkli@mun.ca}
\homepage{https://www.math.mun.ca/merkli/}
\orcid{0000-0002-3990-6155}

\maketitle

\begin{abstract}
A spin (qubit) is in contact with a bosonic reservoir. The state of the reservoir contains a parameter $\varepsilon$ interpolating between quantum and classical reservoir features. We derive the explicit expression for the time-dependent reduced spin density matrix, valid for all values of $\varepsilon$ and for energy conserving interactions. We study decoherence and markovianity properties. Our main finding is that the spin decoherence is enhanced (full decoherence) when the spin is coupled to quantum reservoir states while it is dampened (partial decoherence) when coupled to classical reservoir states. The markovianity properties depend in a subtle way on the classicality parameter $\varepsilon$ and on the finer details of the spin-reservoir interaction. We further examine scattering and periodicity properties for energy exchange interactions.
\end{abstract}

\section{Introduction and main results}

One of the paradigmatic open quantum system models is the spin-boson model, which describes a two level quantum system in contact with a bosonic environment \cite{BP, RH, Leggettal,HS,JP,KMS,CFM23, Trushetal, MMQuantum1, MMQuantum2, MMAHP}. 
The two level system represents a spin $\frac12$ or a qubit, or more generally two degrees of freedom of a more complex system which exchange energy and information (correlations) with an external agent. The latter, called the environment, the bath, or the reservoir, consists of bosonic degrees of freedom, that is, of a set of quantum oscillators. In our setup, we consider a continuum of oscillators (a bosonic quantum field), which is used for instance to describe the quantized electromagnetic field, described by creation and annihilation operators $a^*(k)$, $a(k)$, $k\in\rx^3$, satisfying the canonical commutation relations
$$
[a(k),a^*(l)]=\delta(k-l).
$$ 
The interacting spin-boson Hamiltonian is
\begin{equation}
\label{He0}
H(\varepsilon) = \tfrac12\omega_0\sigma_z +\int_{\rx^3}\omega(k) a^*(k)a(k) d^3k + \lambda \sqrt\varepsilon\, G\otimes\varphi(g),
\end{equation}
where $\omega_0>0$ is the spin Bohr frequency and $\omega(k)\ge 0$ is the reservoir dispersion relation. The interaction operator contains a mean-field or quasi-classical parameter $0<\varepsilon\le 1$ which we will explain below. It contains further a coupling constant $\lambda\in\rx$, an operator $G=G^*$ acting on the spin and the field operator
\begin{equation}
\label{fieldop}
\varphi(g) = \frac{1}{\sqrt 2}\big[ a^*(g)+a(g)\big],
\end{equation}
where
$$
a^*(g) = \int_{\rx^3} g(k)a^*(k)d^3k
$$
and $g(k)$ is a (complex-valued) function, called the {\em form factor}. The Schr\"odinger-von Neumann dynamics of an initial density matrix $\rho_{\s\r}$ of the combined system is given by ($\hbar=1$ in our units)
$$
\rho_{\s\r}(t) =e^{-i t H(\varepsilon)}\rho_{\s\r}e^{it H(\varepsilon)}.
$$
In this work, we consider factorized initial states (this is not necessary for our arguments to work though, see \cite{CFO}) 
\begin{equation}
\label{ist}
\rho_{\s\r} = \rho_\s\otimes\zeta_\varepsilon,
\end{equation}
in which the spin density matrix $\rho_\s$ is arbitrary and the reservoir density matrix $\zeta_\varepsilon$ describes a `macroscopic' state, populated with many particles or excitations (infinitely many as $\varepsilon\rightarrow 0$).  Our model has two quasi-classical (mean-field) features, which we discuss
in more detail below:

\begin{itemize}
\item[(QC1)] The average number of particles or excitations in the reservoir state is `macroscopically' large $\propto\varepsilon^{-1}$ where $\varepsilon\rightarrow 0$. The reservoir observables are of mean-field type, that is, $n-$body reservoir observables are rescaled with a prefactor $\varepsilon^n$ so that their averages are of order $1$ as $\varepsilon\rightarrow 0$. Due to this scaling, the average of a commutator of two $n$-body mean-field operators is of the order $\varepsilon$. In the limit $\varepsilon\rightarrow 0$ the reservoir becomes classical, in the sense that its observables commute. 
\item[(QC2)] The spin-reservoir interaction term is scaled in the mean-field sense (factor $\sqrt\varepsilon$ in \eqref{He0}), so that the interaction operator $\lambda\sqrt\varepsilon G\otimes\varphi(g)$ is of the same order as the spin energy $\tfrac12\omega_0\sigma_z$,  independently of the value of $\varepsilon$. 
\end{itemize}
\medskip

We present a discussion of these two quasi-classical features in Section \ref{features}.

\subsection{Outline of the main results}

Our main results are on the classical limit ($\varepsilon\rightarrow 0$) of reservoir states (without coupling to a system) and on the effective dynamics of the spin coupled to the classical versus quantum reservoir states (that is, various values of $\varepsilon$).

\begin{itemize}
\item[1.] {\em Classical limit of reservoir states.} \\
We calculate the explicit form of the characteristic function in the classical limit, for the bosonic reservoir in a coherent state, in a Bose-Einstein (BE) condensate and in a thermal equilibrium state. We show that in the classical limit, the BE condensate state is a uniform  mixture of coherent states varying over an angle $\theta\in[-\pi,\pi]$. We show that the classical limit of the equilibrium state is obtained as a high-temperature limit of the corresponding quantum state. 

\item[2.] {\em Spin dynamics for energy-conserving interaction.} \\
We derive the exact dynamics of the spin density matrix for $0\le \varepsilon\le 1$, when the interaction operator $G$ in \eqref{He0} commutes with $\sigma_z$ (energy conserving model). 

\begin{itemize}
\item[(a)] Our main result on {\em decoherence} is summarized as follows: When the spin is coupled to the reservoir in a quantum coherent state or a quantum BE condensate state ($\varepsilon >0$) then the spin undergoes full decoherence, meaning that its off-diagonal density matrix ($\sigma_z$ basis) converges to zero in the limit of large times. However, when the same spin is coupled to the reservoir in the classical limit state ($\varepsilon =0$), then the spin only exhibits partial decoherence (nonzero asymptotic value of the off-diagonal density matrix element). The quantum nature of the reservoir thus enhances the decoherence of the spin. We also show that the spin shows full decoherence when coupled to the thermal reservoir, both in the classical and the quantum case.  This can be viewed as
a consequence of the classical limit being equivalent to a high temperature quantum
case.

\item[(b)] Our main results on {\em (non-)Markovianity} are given by numerical simulations. We find that the dynamics of the spin coupled to the condensate/coherent/thermal state is Markovian for small times, regardless of the value of the classicality parameter $\varepsilon$, and subsequently oscillates in time between Markovian and non-Markovian regimes. We find that strong coupling of the infra-red modes favours Markovianity. We also find
that a larger degree of quantumness (larger values of $\varepsilon$) in the reservoir favours non-Markovianity (thermal case).
\end{itemize}

\item[3.] {\em Spin dynamics for the energy exchange interaction.}\\
We show that our exact results for the quasi-classical spin dynamics in the energy-conserving model coincides with the expression derived in a more general theory of quasi-classical dynamics \cite{CFO}. The latter applies as well to the quasi-classical limit for interactions which do not conserve the spin energy ($G$ does not commute with $\sigma_z$). Our results for this case, which is not explicitly solvable, are as follows.
\begin{itemize}
\item[(a)] We show that the scattering operator for the system-reservoir dynamics in the quasi-classical limit exist, provided a certain `dispersiveness' condition on the reservoir dynamics is satisfied. The latter is encoded by the infra-red and ultra-violet behaviour of the interaction (form factor $g(k)$). As a consequence, the free spin dynamics, generated by the Hamiltonian $\tfrac12\omega_0 \sigma_z$, is stable in the sense that it deviates from the interacting dynamics ($\lambda\neq 0$) by a quantity of $O(\lambda)$ for all times $t\ge 0$. 

\item[(b)] While the energy-conserving model with the classical coherent state or the classical Bose-Einstein condensate state falls within the situation of the stable free spin dynamics (see point 2(a) above -- even though the analysis there holds for all values of the coupling $\lambda$), the thermal case does not. We present another family of models for which the free spin dynamics is not stable: The polaron-type models. For those the field dispersion relation is $\omega(k)=\omega_\r$, a constant in $k$, which leads to a violation of the dispersiveness condition mentioned above. We show that the effective spin dynamics is governed by a time-dependent Hamiltonian which is well known from the theory of atoms interacting with classical electric fields. 
\end{itemize}
\end{itemize}

\subsection{Quasi-classical regime.}
\label{features}
 We give a discussion of the features (QC1) and (QC2).
\smallskip

$\bullet$ {\it Discussion of {\rm (QC1)}.} This is a property of the reservoir alone. To motivate it, consider for the moment the reservoir without any coupling to the spin. Assume there are $N$ particles in the reservoir,  where $N$ is a fixed number. Pure states are symmetrized functions of $N$ variables, that is, normalized elements of 
$$
L^2_{\rm sym}(\rx^{3n},d^{3n}k) = \mathcal S\,  L^2(\rx^3,d^3k)\otimes\cdots\otimes L^2(\rx^3,d^3k),
$$
where $\mathcal S$ is the symmetrization operator. Let $A_1$ be a single particle operator acting on $L^2(\rx^3, d^3k)$ with integral kernel $A_1(k,l)$,
$$
(A_1f)(k)=\int_{\rx^3} A_1(k,l)f(l)d^3l.
$$
The associated one-body operator is defined by
\begin{eqnarray*}
\mathcal O(A_1) &=& \int_{\rx^3\times\rx^3}  A_1(k,l) a^*(k)a(l)\,d^3kd^3l\\
&=& A_1\otimes\bbbone\otimes\cdots\otimes\bbbone + \bbbone\otimes A_1\otimes\bbbone\cdots\otimes\bbbone +\cdots + \bbbone\otimes\cdots\bbbone\otimes A_1.
\end{eqnarray*}
This operator acts equally on each particle in an $N$-body system and it adds up the corresponding actions over all $N$ particles. It    preserves the symmetry of states. In {\em mean field theories}, one considers observables {\em averaged} over the number of particles. One-body mean field operators are defined to be given by $N^{-1}\mathcal O(A_1)$. In the same way, $n$-body operators ($1\le n\le N$) are of the form
$$
\mathcal O(A_n) = \int_{\rx^{3n}\times\rx^{3n}} A_n(k_1,\ldots,l_n)\, a^*(k_1)\cdots a^*(k_n)   a(l_1)\cdots a(l_n) d^3k_1\cdots d^3l_n,
$$
with a kernel $A_n(k_1,\ldots,l_n)$. The operator $\mathcal O(A_n)$ selects each of the  $N\choose n$ possible clusters of $n$ particles, acts on each cluster in the same way dictated by $A_n$, and then adds up all those actions. For large $N$ (and fixed $n$), we have 
$$
{N\choose n} \sim N^n
$$
and mean field $n$-body operators are defined to be of the form $N^{-n}\mathcal O(A_n)$. The mean field scaling is implemented by the replacement 
$$
a^*(k)\mapsto N^{-1/2} a^*(k),\quad a(k)\mapsto N^{-1/2} a(k).
$$
The expectation values of the rescaled creation and annihilation operators in $N$-body
states are of order $1$ as $N\rightarrow\infty$. 

The above is the description of the reservoir and its mean-field observables for a fixed number $N$ of particles. However, in an open system setting, the reservoir is coupled to a spin, and absorption/emission processes will change that number. Nevertheless, typically the interaction with the spin can alter the average number of particles in the bath by a finite amount only, and so during the dynamics, the number of excitations in the bath stays of the initial order $\varepsilon^{-1}$.  Let 
$$
\widehat N=\int_{\rx^3} a^*(k)a(k)d^3k
$$
be the number operator of the bath and consider a bath density matrix $\zeta_\varepsilon$, such that the average number of particles is given by
\begin{equation}
\label{12}
{\rm Tr}_\r\big(\zeta_\varepsilon \widehat N)\propto \varepsilon^{-1}.
\end{equation}
The mean-field scaling explained above amounts to the replacement
$$
a^*(k)\mapsto a_\varepsilon(k):= \sqrt\varepsilon a^*(k)
$$
(and analogously for $a(k)$) in the expression for reservoir observables, or analogously, to the scaling
\begin{eqnarray*}
a_\varepsilon^*(f):= a^*(\sqrt\varepsilon f), && a_\varepsilon(f):= a(\sqrt\varepsilon f)\\
\varphi_\varepsilon(f)&:=& \varphi(\sqrt\varepsilon f)\\   W_\varepsilon(f)&:=& W(\sqrt\varepsilon f):=e^{i\varphi_\varepsilon(f)}.
\end{eqnarray*}
For any fixed $0<\varepsilon\le1$, 
${\rm Tr}_\r (\zeta_\varepsilon \,\cdot\,)$ is a state (positive linear normalized functional) on the Weyl algebra $\mathcal W$, the $C^*$-algebra generated by all the operators $W(f)$, $f\in L^2(\rx^3,d^3k)$. This is the same as the $C^*$-algebra generated by all the $W_\varepsilon(f)$, of course. However, as we take $\varepsilon\rightarrow 0$, we will focus on the mean-field observables $W_\varepsilon(f)$ exclusively. Namely, we will consider the limit
$$
\chi_\varepsilon(f)={\rm Tr}_\r\big(\zeta_\varepsilon W_\varepsilon(f)\big),\quad \mbox{as $\varepsilon\rightarrow 0$.}
$$
The limiting characteristic functional 
$$
\chi_0(f)=\lim_{\varepsilon\rightarrow 0}\chi_\varepsilon(f)
$$
then defines a state on $\mathcal W$ which is obtained as a limit of averages of mean-field observables only. This restriction naturally imposes some macroscopic, classical features on the limit state associated to $\chi_0$. 
The emergence of classicality can be explained as follows. Let $A_\varepsilon, B_\varepsilon$ be reservoir observables, that is,  polynomials (and limits thereof) in creation and annihilation operators $a^*_\varepsilon(f)$, $a_\varepsilon(f)$. Then we have the commutative (classical) property 
$$
\lim_{\varepsilon\rightarrow 0} {\rm Tr}_\r\big(\zeta_\varepsilon A_\varepsilon B_\varepsilon\big) = \lim_{\varepsilon\rightarrow 0} {\rm Tr}_\r\big(\zeta_\varepsilon B_\varepsilon A_\varepsilon\big),
$$
which readily follows from $[A_\varepsilon,B_\varepsilon]=O(\varepsilon)$. Of course, $A_\varepsilon, B_\varepsilon\rightarrow 0$ in the limit $\varepsilon\rightarrow 0$, but the state $\zeta_\varepsilon$ is also scaled with $\varepsilon$ in order to have generally non-trivial (nonzero) expectation values as $\varepsilon\rightarrow 0$. 

We close the discussion of (QC1) by a comment on the dynamics of the reservoir. 
All reservoir observables, mean-field or not, evolve according to the Bogliubov transformation $f\mapsto e^{i\omega t}f$ implemented by the field Hamiltonian $$
H_\r=\int_{\rx^3}\omega(k)a^*(k)a(k)d^3k.
$$
The scaling $\varepsilon H_\r$ represents the mean-field energy observable, that is, the average energy per particle. However, as the generator of the dynamics, $H_\r$ is not scaled with $\varepsilon$. This is natural since the dynamics is generated by the full energy operator, not by
its mean-field counterpart, and also since this choice provides a non-trivial effective
dynamics for both the spin and the field on the same time scale of the order $1$ (that is,
$O(\varepsilon^0)$).
\medskip 

$\bullet$ {\it Discussion of {\rm (QC2)}.} The coupling constant $\lambda$ in the interaction term is independent of the parameter $\varepsilon$. The latter is linked to the initial state of the reservoir, but $\lambda$ is not. The quasi-classical limit $\varepsilon\rightarrow 0$ is {\em not} the same as a weak coupling limit $\lambda\rightarrow 0$, because the reservoir state scales singularly with $\varepsilon\rightarrow 0$, and the convergence $\sqrt\varepsilon\varphi(g)\rightarrow 0$ is compensated by the divergence of the state. As a consequence, the interaction felt by the spin is not small as $\varepsilon\rightarrow 0$. Rather, in this limit, the spin feels the interaction with a classical reservoir, whose (mean-field scaled) observables commute.
\medskip

In our model only the reservoir part is scaled to become classical as $\varepsilon\rightarrow 0$ while the spin  does not undergo such a scaling. We call this the {\em quasi-classical limit} \cite{CFO,CFO19,CF18,CFO23,CFM23} (as opposed to the semi-classical one, in which the system would become classical in the limit as well). The coupling of a system of quantum particles in interaction with a field with a scaling regime analogous to ours was considered first, to the best of our knowledge, in \cite{GNV06}. Further works using state-valued Wigner measures are \cite{CF18, CFO19, CCFO21, CFO, CFM23, CFO23}, while in \cite{amour2016anarxiv,amour2016amrex,amour2017jmp,amour2017arxiv,amour2016alnarxiv} and references therein a different approach is taken, using infinite-dimensional Weyl calculus.

\section{Classical limit of the quantum reservoir}

In this section we consider the reservoir alone. We describe its Hilbert space and introduce the characteristic function of reservoir states in Section \ref{secCF}. We then derive the explicit characteristic functions for coherent, condensate and thermal states, Section \ref{sec:classlim}. 

The reservoir Hilbert space is the bosonic Fock space 
$$
\mathcal F(\h) := \bigoplus_{n\ge 0} \big[\h^{\otimes n}\big]_{\rm sym}  = \cx\oplus \h \oplus [\h\otimes\h]_{\rm sym}\oplus [\h\otimes\h\otimes\h]_{\rm sym}\oplus\cdots
$$
with the single-particle space (in momentum representation)
$$
\h = L^2(\rx^3,d^3k).
$$
The symbol $[\cdot]_{\rm sym}$ means that we take the symmetric (permutation invariant) subspace. The $n$-fold symmetric product $[\h\otimes\cdots\otimes\h]_{\rm sym}$ is called the $n$-particle sector. The creation, annihilation, field and Weyl opeators are denoted by $a^*(f)$, $a(f)$, $\varphi(f)$ (as above) and \begin{equation}
\label{Weylop}
W(f)=e^{i\varphi(f)}.
\end{equation}
They satisfy the canonical commutation relations (equivalently expresseed in three ways),
\begin{equation}
\label{CCR}
[a(f),a^*(g)] = \langle f,g\rangle,\quad [\varphi(f),\varphi(g)] =i{\rm Im} \langle f,g\rangle,\quad 
 W(f) W(g) = e^{-\frac{i}{2}{\rm Im}\langle f,g\rangle} W(f+g),
\end{equation}
where $\langle \cdot\, , \cdot\rangle$ is the inner product of $\h$.

\subsection{Characteristic function of the reservoir}
\label{secCF}

The {\em characteristic function} of a reservoir state $\zeta_\varepsilon$ (a density matrix on $\mathcal F(\h)$) is defined to be 
\begin{equation}
\label{defcharfn}
\chi_\varepsilon(f) ={\rm Tr}_\r(\zeta_\varepsilon W_\varepsilon(f)), \qquad f\in\h.
\end{equation}  
We consider reservoir states $\zeta_\varepsilon$ for which the following two assumptions hold:
\begin{itemize}
\item[1.] The limit $\varepsilon\rightarrow 0$ of the characteristic function $\chi_\varepsilon(f)$ exists for every $f\in\h$ and we denote it by $\chi_0$, 
\begin{equation}
\chi_0(f) \equiv \lim_{\varepsilon\rightarrow 0}\chi_\varepsilon(f).
\end{equation}

\item[2.] The nonlinear functional  $f\mapsto\chi_0(f)$ is continuous (as a map from $\h$ to $\mathbb C$).
\end{itemize}
\bigskip

\noindent
The functional $\chi_0$ is of positive type, meaning that for any $N\in\mathbb N$, any $z_k\in\cx$, $f_k\in\h$, $k=1,\ldots,N$, we have 
\begin{equation}
\label{21.1}
\sum_{k,l}  \overline{z_k}z_l \, \chi_0 (f_k-f_l)\ge 0.
\end{equation}
To see that \eqref{21.1} holds, we note that for any $\varepsilon>0$,  $\tr_\r(\zeta_\varepsilon \cdot)$ is a state on the CCR algebra generated by the Weyl operators $W(f)$, $f\in\h$, so we have (see for instance \cite{MMLecNotes, Petz})
$$
\sum_{k,l}\overline{z_k}z_l \, \tr_\r \big(\zeta_\varepsilon W(f_k-f_l)\big) e^{\frac i2  {\rm Im}\langle f_k,f_l\rangle} \ge 0,
$$
which implies that (take $\sqrt\varepsilon f_k$ instead of $f_k$)
$$
\sum_{k,l} \overline{z_k} z_l\, \tr_\r \big(\zeta_\varepsilon W_\varepsilon(f_k-f_l)\big) e^{\varepsilon \frac i2  {\rm Im}\langle f_k,f_l\rangle} \ge 0.
$$
Taking $\varepsilon\rightarrow 0$ yields \eqref{21.1}. Furthermore, we have $\chi_\varepsilon(0)=1$ for all $\varepsilon>0$  and so 
\begin{equation}
\label{21.2}
\chi_0(0)=1.
\end{equation} 
The continuity in $f$ and the properties \eqref{21.1} and \eqref{21.2} imply, by the Bochner-Minlos theorem \cite{Hida,Fa}, that $\chi_0$ is the Fourier transform of a {\em cylindrical} probability measure $\mu$ on $\h$,\footnote{A cylindrical measure or premeasure, which in general is not a measure, is a finitely additive set function defined on so-called cylinder sets. We refrain from giving technical definitions here and refer to  \cite{cylmeas, skorohod, Bogachev} for detail. The Bochner-Minlos theorem is often stated for real Hilbert spaces. The complex Hilbert space $\h$ with orthonormal basis $\{e_n\}_{n\ge 0}$ and inner product $\langle\cdot,\cdot\rangle$ is a real Hilbert space with orthonormal basis $\{e_n, ie_n\}_{n\ge 0}$ and inner product ${\rm Re}\langle\cdot,\cdot\rangle$. 
}
\begin{equation}
\label{16}
\chi_0(f)=\int_\h d\mu(g) e^{i\sqrt 2 {\rm Re}\langle g,f\rangle}.
\end{equation}
Cylindrical (Wigner) measures were first used to study bosonic mean-field theory  in \cite{AmNi1} and subsequently in  \cite{AmNi2,AmNi3,AmNi4} for a very general setting. It was shown that if the state $\zeta_\varepsilon$ satisfies the additional assumption that for some $\delta>0$ and some $C_\delta<\infty$ we have 
\begin{equation}
\label{A1}
{\rm Tr}_\r\big(\zeta_\varepsilon (\varepsilon\widehat N +1)^\delta\big) <C_\delta,
\end{equation}
then the measure $\mu$ in \eqref{16} is a (true, that is $\sigma$-additive) probability measure on $\h$. The condition \eqref{A1} is satisfied for coherent and condensate states we consider in this work and is compatible with the setting described above, {\em i.e.}, an initial state of the field containing an average number of excitations of the order $ O(\varepsilon^{-1}) $. Mathematically speaking, \eqref{A1} ensures that the state essentially concentrates on the sectors of the (unscaled) Fock space with finitely many excitations | that is, its tails on sectors with very many excitations decay fast enough. This property guarantees that the associated cylindrical measure is what is called {\it tight} and, as such, it identifies a proper measure on the Fock space.

The inclusion of the factor $\sqrt 2$ in the phase of \eqref{16} is a convention (that factor can be chosen to be any nonzero real number; its specific choice determines the probability measure $\mu$). As we will see below in Section \ref{sec:cs}, with the choice $\sqrt 2$ in the exponent, the measure $\mu$ in \eqref{16} has a direct interpretation for coherent reservoir states.

\subsection{Explicit classical limits of some reservoir states}
\label{sec:classlim}

\subsubsection{Coherent states} 
\label{sec:cs}

Coherent states are indexed by $f\in \h$, defined  as 
\begin{equation}
\label{coherentstate}
\Psi_f := e^{a^*(f)-a(f)}\Omega = e^{i\varphi(-\sqrt 2 i f)}\Omega = W(-\sqrt2 i f)\Omega,
\end{equation}
where $\Omega$ is the vacuum state and where we used the definitions \eqref{fieldop} and \eqref{Weylop} to arrive at the second and third equality, respectively. We consider reservoir density matrices of the form
\begin{equation}
\label{m13}
\zeta = \int_\h d\mu_0(f) |\Psi_f\rangle\langle\Psi_f|,
\end{equation}
where $d\mu_0$ is a probability measure on $\h$. This state is generally a mixture of the pure coherent states $|\Psi_f\rangle\langle\Psi_f|$. For $d\mu_0(f) = \delta_{f,f_0}$ (Dirac measure centered at a fixed $f_0\in\h$) we have $\zeta=|\Psi_{f_0}\rangle\langle \Psi_{f_0}|$. The average number of particles in the state $\zeta$ is 
$$
{\rm tr}_\r[\zeta \widehat N] = \int_\h d\mu_0(f) \langle \Psi_f, \widehat N \Psi_f\rangle  = 2 \int_\h d\mu_0(f) \|f\|^2.
$$
In accordance with the classical scaling we set for $\varepsilon>0$,
\begin{equation}
\label{m13.1}
\zeta_\varepsilon = \int_\h d\mu_0(f) |\Psi_{f/\sqrt\varepsilon }\rangle\langle\Psi_{f/\sqrt \varepsilon }|,
\end{equation}
so that $\tr_\r[\zeta_\varepsilon \widehat N]\propto 1/\varepsilon$. A direct calculation based on the Weyl CCR \eqref{CCR}, the definition \eqref{coherentstate} and the fact that $\langle \Omega, W(f)\Omega\rangle =e^{-\tfrac14 \|f\|^2}$ yields
\begin{equation}
\label{chicoherent}
\chi_\varepsilon(f) = \int_{\h}d\mu_0(g)\,  {\rm tr}_\r\big[|\Psi_{g/\sqrt\varepsilon }\rangle\langle\Psi_{g/\sqrt \varepsilon }| W_\varepsilon(f)\big] =   e^{-\frac14\varepsilon\|f\|^2}\ \int_\h d\mu_0(g)\,  e^{i \sqrt2  {\rm Re}\langle g,f\rangle}.
\end{equation}
Thus
\begin{equation}
\label{chinotcoherent}
\chi_0(f) = \lim_{\varepsilon\rightarrow 0} \chi_\varepsilon(f) = \int_\h d\mu_0(g)\,  e^{i \sqrt2  {\rm Re}\langle g,f\rangle}
\end{equation}
and so by comparing with \eqref{16} we see that the measure resulting from the Bochner-Minlos theorem is simply $d\mu_0$.

\subsubsection{Bose-Einstein condensate}
\label{sect:BEC}

Let $f_0\in \h$, $\|f_0\|=1$, be a single particle wave function and consider the $n$-particle state
\begin{equation}
\label{condpurestate}
\psi_\varepsilon  = \frac{a^*(f_0)^n}{\sqrt{n!}}\Omega, \quad n=\lfloor 1/\varepsilon \rfloor
\end{equation}
where $\Omega$ is the vacuum state and $n$ is the largest integer $\le 1/\varepsilon$. This scaling ensures that the number of particles in the state $\psi_\varepsilon$ is $n\sim 1/\varepsilon$, in compliance with  \eqref{12}. The characteristic functional 
\begin{equation}
\label{74.0}
\chi_\varepsilon(f) = \langle\psi_\varepsilon,  W_\varepsilon(f)\psi_\varepsilon\rangle, \qquad f\in \h
\end{equation}
can be calculated explicitly,
\begin{equation}
\label{74}
\chi_\varepsilon(f) = L_n\big(\tfrac12 \varepsilon|\langle f_0,f\rangle|^2\big)\, \langle \Omega, W_\varepsilon(f)\Omega\rangle, \qquad n=\lfloor 1/\varepsilon\rfloor,
\end{equation}
where
$$
L_n(x) = \frac{1}{n!}\sum_{k=0}^n {n\choose k} \frac{n!}{(n-k)!}( -x)^{n-k} \ 
$$
is the $n$th Laguerre polynomial. 
The expression \eqref{74} is derived in \cite{MMLecNotes} and appeared in the work of Araki and Woods \cite{AW}. In the setup of \cite{AW,MMLecNotes} one is interested in the thermodynamic (infinite volume) limit while in the present setting, we take the classical limit. 
The two cases are formally similar but there is one crucial difference: The factor $\langle \Omega, W_\varepsilon(f)\Omega\rangle$ in \eqref{74} converges to $1$ as $\varepsilon\rightarrow 0$ while that factor remains present in the infinite volume limit. Due to the disappearence of the factor, the represented algebra of observables is commutative in the classical limit (see below), while it is not commutative in the infinite volume limit. As derived in \cite{AW,MMLecNotes}, we have
\begin{equation}
\label{76}
\lim_{\varepsilon\rightarrow 0} \chi_\varepsilon(f) = J_0\big(\sqrt{2}|\langle f_0,f\rangle| \big) = \int_{-\pi}^\pi  \frac{d\theta}{2\pi} \ e^{i\sqrt 2\, {\rm Re}\,  e^{i\theta} \langle f_0,f\rangle}, 
\end{equation}
where $J_0$ denotes the Bessel function and the second equality is a well known integral representation of $J_0$. According to \eqref{16}, the measure $d\mu$ associated to the quasiclassical reservoir state is supported on the states $e^{i\theta}f_0$ with $\theta$ drawn uniformly from $S^1$.

To recover a Hilbert space formalism, we may cast \eqref{76} in the form
$$
\lim_{\varepsilon\rightarrow 0}\langle\psi_\varepsilon, W_\varepsilon(f)\psi_\varepsilon\rangle = \langle \Omega_{\rm cl}, W^{\rm cl}(f)\Omega_{\rm cl}\rangle,
$$
where the Hilbert space on the right side is $L^2(S^1,d\theta/2\pi)$, $\Omega_{\rm cl}=1$ is the constant function and the `classical Weyl operator' is the operator of multiplication by the function $W^{\rm cl}(f) = e^{i\sqrt2\, {\rm Re} \, e^{i\theta}\langle f_0,f\rangle}$ on $L^2(S^1,d\theta/2\pi)$. 

The classical Weyl operators {\em commute} (they generate the bounded multiplication operators on  $L^2(S^1,d\theta/2\pi)$). They do not satisfy the usual canonical commutation relations. This results because with the scaling $W_\varepsilon(f)$, we took the limit $\varepsilon\rightarrow 0$ of the elements in the algebra of observables (not only of the state as it is the case, for example, in the infinite volume limit). The classical field and creation operators are operators of multiplication by
$$
\varphi^{\rm cl}(f) = \sqrt2 {\rm Re}\, e^{i\theta} \langle f_0,f\rangle, \quad (a^*)^{\rm cl}(f) = e^{i\theta}\langle f_0,f\rangle.
$$

The expression \eqref{76}  was  derived by a direct computation in \cite[proposition 4.1]{AmNi1}, without linking it to the classical representation of the Weyl algebra on a Hilbert space. 

\subsubsection{Link between coherent and Bose-Einstein states}
\label{subsub:link}
The relations \eqref{chinotcoherent} and \eqref{76} show that in the classical limit, the Bose-Einstein condensate populated by particles with wave function $f_0\in\h$ is the same reservoir state as the one obtained by mixing the family of coherent states  $\Psi_{e^{-i\theta} f_0}$, \eqref{coherentstate}, uniformly over $\theta\in[-\pi,\pi]$.

\subsubsection{Thermal state}
\label{sec:thermal}

 Without the classical scaling, the thermal state is characterized by the two-point function 
\begin{equation}
\label{m20}
\langle a^*(f)a(g)\rangle_\beta = \int_{\rx^3} \frac{\overline g(k) f(k)}{e^{\beta\omega(k)}-1}d^3k.
\end{equation}
Here, $\beta$ is the inverse temperature and the free field Hamiltonian is $\d\Gamma(\omega)$. The symbol $\langle \cdot \rangle_\beta$ denotes the average in the thermal state. In the classical limiting procedure, the observables are built from $a_\varepsilon(f)=\sqrt\varepsilon a(f)$ and $a^*_\varepsilon(f) =\sqrt{\varepsilon}a^*(f)$, so \eqref{m20} becomes
\begin{equation}
\label{m21}
\langle a_\varepsilon^*(f)a_\varepsilon(g)\rangle_{\beta,\varepsilon} = \int_{\rx^3} \frac{\varepsilon} {e^{\beta\omega(k)}-1}\overline g(k) f(k)d^3k.
\end{equation}
In order to obtain a non-trivial  limit for $\varepsilon\rightarrow 0$ we scale the inverse temperature $\beta$ with $\varepsilon$ as 
$$
\beta(\varepsilon)=\varepsilon \beta',
$$ 
for some fixed $\beta'>0$. This amounts to taking a high-temperature limit in the state, simultaneously with the classical scaling of the observables. The classical thermal state, obtained from \eqref{m21} by taking $\varepsilon\rightarrow 0$, has the two-point function
\begin{equation}
\lim_{\varepsilon\rightarrow 0} \langle a_\varepsilon^*(f)a_\varepsilon(g)\rangle_{\beta'\varepsilon ,\varepsilon} =\langle g,\tfrac{1}{\beta'\omega}f\rangle.
\end{equation}
This is the lowest order term in the high temperature expansion (about $\beta=\beta'= 0)$ of the unscaled (quantum, $\varepsilon=1$) case \eqref{m20}. Taking the limit $\varepsilon\rightarrow 0$ of the characteristic function 
\begin{equation}
\label{chithermal}
\chi_\varepsilon(f) = \langle W_\varepsilon(f)\rangle_{\beta'\varepsilon,\varepsilon} = e^{-\frac14\varepsilon \langle f, \coth(\beta'\varepsilon\omega/2) f\rangle}
\end{equation}
yields
\begin{equation}
\chi_0(f) = e^{-\frac{1}{2\beta'}\langle f, \omega^{-1}f\rangle}.
\end{equation}
The associated measure $d\mu$ (cf \eqref{16}) is the centered Gaussian with covariance operator $\beta'\omega$.

\section{Open spin-reservoir complex}

We now consider a spin coupled to a reservoir. In the quasi-classical theory which we analyze in the present work, the reservoir (Bose field) is considered to be a classical system ($\varepsilon\rightarrow 0$ as above), while the system (spin) stays a quantum object. This is in contrast to the `semi-classical limit', where both the reservoir and the spin become classical systems. The pure state space of the spin is \begin{equation}
\H_\s=\cx^2.
\end{equation}
System observables are selfadjoint operators on $\H_\s$, such as the Pauli matrices $\sigma_x, \sigma_y, \sigma_z$. The spin dynamics is generated by the Hamiltonian
\begin{equation}
\label{HS}
H_\s = \tfrac12\omega_0 \sigma_z,\qquad 
\end{equation}
where $\omega_0>0$.

The full system-reservoir Hamiltonian acting on $\H_\s\otimes\mathcal F(\h)$ is
\begin{eqnarray}
\label{Hamilt}
H&=& H_0+ \lambda G\otimes \varphi(g)\\
H_0&=&H_\s\otimes\bbbone_\r+\bbbone_\s\otimes H_\r = \tfrac12\omega_0 \sigma_z +d\Gamma(\omega),
\end{eqnarray}
where we leave out obvious factors $\otimes\bbbone_\r$ etc. The parameter $\lambda\in\rx$ in \eqref{Hamilt} is the coupling constant, $G=G^*$ is a matrix on $\H_\s$ and $g\in\h$ is a fixed function, called the form factor. For $\lambda=0$ the Hamiltonian $H$ reduces to the uncoupled $H_0$.

In the quasi-classical scaling, the interacting Hamiltonian is
\begin{equation}
\label{Heps}
H(\varepsilon) = H_0+\lambda G\otimes\varphi_\varepsilon(g). 
\end{equation}
For $\varepsilon=1$ the model is fully quantum and for $\varepsilon\rightarrow 0$ we get the quasi-classical model.
\medskip

A standard way of ensuring that $H$, \eqref{Hamilt} is a well defined (selfadoint) operator is to assume that $g\in L^2(\mathbb R^3,d^3k)$. (For the thermal case, we additionally assume that $\omega^{-1/2}g\in L^2(\mathbb R^3,d^3k)$ so that \eqref{m20} is well defined.) However, it is possible to make sense of $H$ in some cases even if $g$ is not square integrable, because $g$ might have  singular infra-red or ultra-violet behaviour. This has been explored rigorously in the recent works \cite{Takaesu2010, Lonigro22, Lonigro23, BCFF24} (UV singular models), where many further references on this topic can be found. To our knowledge, the infrared singular problem has been studied mostly  concerning the existence of the ground and scattering states, \cite{HHS24, BDG24}. As we have explained in result (2b) above, the coupling to the infra-red modes has an effect on the system's Markovianity and it would be interesting to study this effect for singular coupling functions $g$. We have not addressed this question so far.

\subsection{Reduced dynamics of the spin}
\label{sect:qcCFO}

A rigorous theory for the reduced dynamics of a system coupled to the reservoir in the classical limit has been carried out in \cite{CFO}. The treatment includes classes of infinite-dimensional systems as well as  correlated initial system-reservoir states. We apply the results of \cite{CFO} to the relatively simple spin-Boson model in the quasi-classical limit, and for initially uncorrelated system-reservoir states.

Take an initial product state
\begin{equation}
\label{m2}
\Gamma_\varepsilon = \gamma\otimes \zeta_\varepsilon,
\end{equation}
where $\gamma$ is a density matrix on $\cx^2$ and $\zeta_\varepsilon$ is a density matrix on Fock space $\mathcal F(\h)$ having the properties presented in \ref{secCF}. The reduced system density matrix at time $t$ is given by
\begin{equation}
\gamma_\varepsilon(t) = \tr_\r \big[ e^{-i t H(\varepsilon)}(\gamma\otimes \zeta_\varepsilon)e^{i t H(\varepsilon)}\big].
\end{equation}
The following result is a consequence of a more general analysis of quasi-classical dynamics.

\begin{thm}[Correggi-Falconi-Olivieri \cite{CFO}] 
\label{thm:CFO}
For all $t\in\rx$, we have 
\begin{equation}
\label{m18}
\lim_{\varepsilon\rightarrow 0}\gamma_\varepsilon(t) = \gamma_0(t) := 
\int_\h d\mu(f) U_t(f)\gamma U_t(f)^*,
\end{equation}
where $d\mu$ is the probability measure on $\h$ arising from the limit characteristic function $\chi_0$, as given in \eqref{16}. Here,  $U_t(f)\equiv U_{t,0}(f)$ where $U_{t,s}(f)$ is the unitary two-parameter group acting on  $\H_\s=\cx^2$, solving the Schr\"odinger equation with a time-dependent Hamiltonian,
\begin{equation}
\label{m21.1}
i \partial_t U_{t,s}(f) = \big(H_\s+ V(e^{-it\omega} f)\big) U_{t,s}(f),\qquad U_{t,t}(f)=\bbbone,
\end{equation}
with
\begin{equation}
\label{m21.2}
V(f) = \sqrt 2\lambda  
G \,{\rm Re} \langle  f,g\rangle,
\end{equation}
and where $G$ and $g$ are the ingredients of the interaction operator \eqref{Hamilt}. 
\end{thm}

{\em Notational convention:} In \cite{CFO} the field operator $\varphi(g)$ is defined as $\sqrt 2$ times our field operator \eqref{fieldop}, which amounts to replacing our form factor $g$ in \eqref{Heps} by $g/\sqrt 2$ in order to obtain the formulas in \cite{CFO}.

\section{The energy conserving model}
\label{sec:econmodel}

The case where the interaction operator $G$ in \eqref{Hamilt} commutes with the spin Hamiltonian $H_\s$ \eqref{HS} is called energy conserving. In the current section, we first derive the exact spin dynamics in Proposition \ref{prop1}. In Section \ref{sec:deco} we derive from it the decoherence properties, which are summarized in Corollary \ref{cor:2}. We study in Section \ref{sec:nm} the non-Markovianity of the spin dynamics. Section \ref{sec:bench} is devoted to showing that the exact expressions coincide with the results from \cite{CFO} in the quasi-classical limit.

Consider the energy conserving Hamiltonian
\begin{equation}
\label{m1}
H(\varepsilon)= \tfrac12 \omega_0 \sigma_z+\d\Gamma(\omega) + \tfrac12 \lambda \sigma_z\otimes \varphi_\varepsilon(g).
\end{equation}
The system-reservoir dynamics generated by \eqref{m1} is explicitly solvable. We denote the matrix elements of any operator $S$ on $\cx^2$ in the energy basis -- the eigenbasis of $\sigma_z$ -- by $S_{ij}\equiv [S]_{ij}$, where $\sigma_z|1\rangle=|1\rangle$, $\sigma_z|2\rangle=-|2\rangle$. Starting in an initial product state of the form \eqref{m2}, the reduced system density matrix at time $t$,
 \begin{equation}
\gamma_\varepsilon(t) = \tr_\r \big[ e^{-i t H(\varepsilon)}(\gamma\otimes \zeta_\varepsilon)e^{i t H(\varepsilon)}\big],
\end{equation}
satisfies the following.
\begin{prop}[Explicit expression of the dynamics]
\label{prop1}
Let $\varepsilon>0$. Then for all $t\ge 0$, the populations of the spin are constant,  $[\gamma_\varepsilon(t)]_{ii} = \gamma_{ii}$, $i=1,2$ and the off-diagonal evolves as 
\begin{equation}
[\gamma_\varepsilon(t)]_{12} = e^{-i \omega_0 t} D_\varepsilon(t) \gamma_{12},
\end{equation}
where the decoherence function is given by
\begin{equation}
\label{m12}
D_\varepsilon(t)  = \chi_\varepsilon\big(\lambda g_t\big),\qquad g_t(k):= \frac{1-e^{i\omega t}}{i\omega}g(k),
\end{equation}
with $\chi_\varepsilon$ the characteristic function \eqref{defcharfn}.
\end{prop}

{\bf Proof.} Since the interaction operator commutes with $\sigma_z\otimes\bbbone$, the diagonal of $\gamma_\varepsilon$ is clearly time-independent. Next, 
\begin{equation}
[\gamma_\varepsilon(t)]_{12} = \tr_\r \big[ e^{-i t H(\varepsilon)}(\gamma\otimes \zeta_\varepsilon)e^{i t H(\varepsilon)}\  (|2\rangle\langle 1|\otimes\bbbone_\r) \big]
= e^{-i \omega_0 t} D_\varepsilon(t) \gamma_{12},
\end{equation}
with
\begin{equation}
\label{m7}
D_\varepsilon(t) = \tr_\r\big[ \zeta_\varepsilon \,  e^{i t [\d\Gamma(\omega) -\frac\lambda2 \varphi_\varepsilon(g)]} e^{-i t [\d\Gamma(\omega) +\frac\lambda2 \varphi_\varepsilon(g)]} \big ] .
\end{equation}
Using the polaron transformation we can rewrite the exponentials in \eqref{m7}. Namely, the relations
\begin{eqnarray*}
W_\varepsilon(f) \d\Gamma(\omega)W_\varepsilon(f)^* &=& \d\Gamma(\omega) -\varphi_\varepsilon(i\omega f)+\tfrac12 \varepsilon\|\sqrt\omega f\|^2_2\\
W_\varepsilon(f)\varphi_\varepsilon(g)W_\varepsilon(f)^* &=& \varphi_\varepsilon(g)-\varepsilon{\rm Im}\langle f,g\rangle,
\end{eqnarray*}
give
\begin{equation*}
W_\varepsilon(f) e^{-it[\d\Gamma(\omega) +\frac\lambda 2\varphi_\varepsilon(g)]}W_\varepsilon(f)^*
= 
e^{-i t [\d\Gamma(\omega) -\varphi_\varepsilon(i\omega f)+\frac\lambda2 \varphi_\varepsilon(g)]} e^{-it\varepsilon[ \frac12 \|\sqrt\omega f\|^2_2-\frac\lambda2 {\rm Im}\langle f,g\rangle ]},
\end{equation*}
so upon choosing
\begin{equation}
\label{m11}
f = \frac{\lambda}{2i\omega}g 
\end{equation}
we get $e^{-it[\d\Gamma(\omega) +\frac\lambda 2\varphi_\varepsilon(g)]}=W_\varepsilon(f)^* e^{-i t \d\Gamma(\omega) }W_\varepsilon(f) \ e^{\frac18 it \varepsilon\lambda^2 \|g/\sqrt\omega\|^2_2}$. 
Using this, and the analogous expression for the first exponential in \eqref{m7}, and also  \eqref{m11}, we arrive at
\begin{eqnarray*}
D_\varepsilon(t) &=& \tr_\r \big[  \zeta_\varepsilon  W_\varepsilon(f) W_\varepsilon(-2e^{i\omega t}f)W_\varepsilon(f) \big]= \tr_\r\big[ \zeta_\varepsilon  W_\varepsilon\big(2[1- e^{i\omega t}]f\big) \big]\nonumber\\
 &=& \tr_\r\big[ \zeta_\varepsilon W_\varepsilon\big(\lambda \tfrac{1- e^{i\omega t}}{i\omega}g\big) \big].
\end{eqnarray*}
This concludes the proof of Proposition \ref{prop1}.\hfill $\blacksquare$
\medskip

Let us now evaluate the decoherence function for the reservoir states presented in Sections \ref{sec:cs}-\ref{sec:thermal}. We denote the quasi-classical decoherence function by
\begin{equation}
D_0(t):=\lim_{\varepsilon\rightarrow 0}D_\varepsilon(t).
\end{equation}
\begin{itemize}
\item[$\bullet$] {\em Coherent state.} The reservoir density matrix is given by \eqref{m13.1} and we obtain from \eqref{m12} and \eqref{chicoherent},
\begin{eqnarray}
D_\varepsilon(t) &=& e^{-\frac14 \varepsilon\lambda^2\| g_t\|^2} \int_\h d\mu_0(f) \, e^{i \sqrt2\lambda  {\rm Re}\langle f, g_t\rangle},
\label{de1}\\
D_0(t)  &=& \int_\h d\mu_0(f) \, e^{i\sqrt2\lambda  {\rm Re}\langle f, g_t\rangle}.
\label{de1.0}
\end{eqnarray}

\item[$\bullet$] {\em Bose-Einstein condensate.} The reservoir density matrix is $\zeta_\varepsilon=|\psi_\varepsilon\rangle\langle\psi_\varepsilon|$, see \eqref{condpurestate},  and we obtain from \eqref{74},
\begin{eqnarray}
\label{52}
D_\varepsilon(t) &=&  e^{-\frac14 \varepsilon\lambda^2\|g_t\|^2}\ L_n\big( \tfrac12 \, \varepsilon\lambda^2\,  | \langle f_0,
g_t \rangle |^2 \big), \qquad n=\lfloor 1/\varepsilon\rfloor,\\
D_0(t) &=& J_0\big( \sqrt2 |\lambda| \ | \langle f_0,g_t \rangle |\big) = \int_{-\pi}^\pi \frac{d\theta}{2\pi} \ e^{i\sqrt 2\lambda\, {\rm Re}\, e^{\i\theta} \langle f_0,g_t\rangle} ,
\label{52.0}
\end{eqnarray}
where $J_0$ is the Bessel function and the second equality is an integral representation of it. 

\item[$\bullet$] {\em Thermal state.}  The characteristic functional of the thermal state at scaled inverse temperature $\beta = \varepsilon\beta'$ is given in \eqref{chithermal} and we obtain from \eqref{m12},
\begin{eqnarray}
D_\varepsilon(t) &=& e^{-\frac14\varepsilon\lambda^2 \langle g_t, \coth(\beta'\varepsilon\omega/2)   g_t\rangle},\label{de3}\\
D_0(t) &=& e^{-\frac{\lambda^2}{2\beta'} \langle g_t,\frac1\omega g_t\rangle}.
\label{44}
\end{eqnarray}
Here we assume that $g$ is such that $\|\omega^{-1/2}g_t\|<\infty$, which imposes the infrared condition, that $\omega(k)^{-1/2}g(k)$ has to be square integrable at $k\sim 0$. For $k\in\rx^3$ and $\omega(k)=|k|$, this means $g(k)\sim |k|^q$ at small $k$, for some $q>-1$. Note that \eqref{de3} reduces to the first factor in \eqref{de1} as $\beta'\rightarrow\infty$. 
\end{itemize}
In all three cases above, the decoherence function cannot exceed unity,
\begin{equation}
    \label{le}
|D_\varepsilon(t)|\le 1,\qquad \forall \varepsilon\ge0, t\ge0.
\end{equation}
This upper bound is seen directly from \eqref{de3} and \eqref{de1} for the thermal and coherent cases. To see that \eqref{le} holds for the condensate case \eqref{52}, we use that $L_n(x)\le e^{x/2}$ for all $x\ge0$ (see \cite{AS},  Eq.~22.14.12), together with $|\langle f_0,g_t\rangle|\le \|f_0\|\, \|g_t\| = \|g_t\|$.
\medskip

\subsection{Decoherence}
\label{sec:deco}

We examine the exponential factor $e^{-\frac14\varepsilon\lambda^2\|g_t\|^2}$ appearing in the quantum decoherence functions \eqref{de1}, \eqref{52} and in a modified way in \eqref{de3}. 
The spectral density of the reservoir is  defined as 
\begin{equation}
\label{specden}
\mathcal J(\omega) = \frac\pi2\omega^2\int_{S^2}|g(\omega,\Sigma)|^2 d\Sigma,\qquad \omega\ge 0,
\end{equation}
where the integral is taken over the unit sphere $S^2\subset \rx^3$, with uniform measure and the function $g(k)$ is represented in spherical coordinates $\rx^3\ni k\leftrightarrow (\omega,\Sigma)\in \rx_+\times S^2$. Consider the photonic dispersion relation $\omega(k)=|k|$. Then from \eqref{m12},
\begin{equation}
\label{int}
\|g_t\|^2 = \frac4\pi  \int_0^\infty \mathcal J(\omega) \frac{1-\cos(\omega t)}{\omega^2} d\omega.
\end{equation}
The behaviour of integrals of the form \eqref{int} as $t\rightarrow\infty$ is linked to the behaviour of $\mathcal J(\omega)$ as $\omega\rightarrow 0_+$. It has been analyzed in connection with decoherence in many works, see for instance \cite{PSE,MM2009,MBS+} and references therein. The most complete analysis, to our knowledge, is carried out in the recent work \cite{Trush} and can be stated as follows.

\begin{thm}[Trushechkin \cite{Trush}]
\label{thm:trush}
Let $S(\omega)$ be an integrable function on $[0,\infty)$ of the form
$$
S(\omega) = \omega^p G(\omega)
$$
for some $p>-1$, where $G(\omega)$ is twice differentiable in an interval $[0,\omega_c)$ 
  for some $\omega_c>0$, $G''(\omega)$ is bounded on $[0,\omega_c)$ and $G(0), G'(0)\neq 0$. Define the function
$$
\Gamma(t) = \int_0^\infty S(\omega)\frac{1-\cos(\omega t)}{\omega^2}d\omega,\qquad t\ge0
$$
and denote by $o(1)$ a function of $t$ which vanishes as $t\rightarrow\infty$. Then the following holds.
\begin{itemize}
\item[1.] If $p>1$  then $\Gamma(t)= \Gamma_\infty:=\int_0^\infty S(\omega)/\omega^2 d\omega +o(1)$.

\item[2.] If $p=1$ then $\Gamma(t)=C_1+S'(0) \ln(t) +o(1)$ for some constant $C_1$.

\item[3.] If $0<p<1$ then $\Gamma(t) = C_2+G(0)C_3 t^{1-p}+o(1)$ for some constants $C_2, C_3$.

\item[4.] If $p=0$ then $\Gamma(t)=C_2+G(0)C_3 t+G'(0)\ln(t)+o(1)$ for the same constants as in point 3.

\item[5.] If $-1<p<0$ then $\Gamma(t) = C_2+G(0)C_4 t^{1-p}+o(t)$ for the same $C_2$ as in point 3.~and for some constant $C_4$.
\end{itemize}
Moreover, if $S(\omega)\ge 0$ (which is the case if $S(\omega)$ is a spectral density), then $\Gamma_\infty>0$, $S'(0)>0$ in point 2.~, $G(0)C_3>0$ in points 3., 4.~and $G(0)C_4>0$ in point 5.  
\end{thm}

Theorem \ref{thm:trush} together with \eqref{int} has the following consequence for the factor $e^{-\frac14 \varepsilon\lambda^2\|g_t\|^2}$ appearing in the decoherence functions (\eqref{de1}, \eqref{52}, \eqref{de3}).

\begin{cor}
\label{cor:deco}
Suppose the reservoir spectral density \eqref{specden} satisfies $\mathcal J(\omega)=\omega^pG(\omega)$ with $G$ as in the statement of Theorem \ref{thm:trush}. 
Then as $t\rightarrow\infty$,
$$
 e^{-\frac14 \varepsilon\lambda^2\|g_t\|^2} \sim
\left\{
\begin{array}{ll}
e^{-\frac1\pi\varepsilon\lambda\Gamma_\infty}, &  p>1,\\
c_0 t^{-\varepsilon\lambda^2 c_1}, & p=1,\\
c_0e^{-\varepsilon\lambda^2c_1 t}\, t^{-\varepsilon\lambda^2 c_2}, & p=0,\\
c_0e^{-\varepsilon\lambda^2c_1 t^{1-p}}, & -1<p<1, p\neq 0,
\end{array}
\right.
$$
for some $c_0>0$ (depending on $\varepsilon\lambda^2$ and converging to $1$ as $\varepsilon\lambda^2\rightarrow 0$) and some $c_1>0$ and $c_2\in\rx$ (both independent of $\varepsilon$, $\lambda$). By $A\sim B$ we mean that $|A-B|\rightarrow 0$ as $t\rightarrow\infty$.
\end{cor}
Corollary \ref{cor:deco} shows that depending on the infrared behaviour of the spectral density of the reservoir, the factor  $e^{-\frac14 \varepsilon\lambda^2\|g_t\|^2}$ in the decoherence functions \eqref{de1}, \eqref{52} is constant, decays polynomially, sub-exponentially, exponentially or super-exponentially, as $t\rightarrow\infty$.  The decoherence functions \eqref{de1}-\eqref{52.0} depend additionally on time via the term ${\rm Re}\langle f,g_t\rangle$. We will assume that
\begin{equation}
\int_{\mathbb R^3} \frac{f(k)g(k)}{\omega(k)} d^3k <\infty 
\end{equation}
for all $f$ in the support of the measure $\mu_0$ (coherent state, \eqref{de1}) and $f=f_0$ (BEC, \eqref{52}). By the Riemann-Lebesgue Lemma the oscillating term in $g_t$ \eqref{m12} vanishes in the limit of large times,
\begin{equation}
\label{49}
\lim_{t\rightarrow\infty} \langle f, g_t\rangle = -i\langle f,\omega^{-1} g\rangle.
\end{equation}

We now state our main observation on decoherence in the quasi-classical versus the quantum case.

\begin{cor}[Spin decoherence from coupling to quantum/classical reservoir]
\label{cor:2}
Suppose the reservoir spectral density behaves as $\mathcal J(\omega)\sim \omega^p$ for small $\omega$, some $p>-1$, as in Theorem~\ref{thm:trush}.
\begin{itemize}
\item[(a)]
When a spin is coupled to a quantum reservoir ($\varepsilon >0$) in a quantum coherent state or a quantum Bose-Einstein condensate, the spin  undergoes full decoherence ($p\le 1$), 
\begin{equation}
\lim_{t\rightarrow\infty} D_\varepsilon(t)=0\qquad \mbox{(coupled to quantum reservoir)}.
\end{equation}
The speed of the decoherence is governed by the value of $p$ according to Corollary \ref{cor:deco}.

\item[(b)] When a spin is coupled to a reservoir in a classical ($\varepsilon =0$) coherent state or a classical Bose-Einstein condensate state, the spin undergoes partial decoherence only, 
\begin{equation}
\lim_{t\rightarrow\infty} D_0(t)=D_0(\infty)\neq 0 \qquad \mbox{(coupled to classical reservoir)}.
\end{equation}

\item[(c)] When a spin is coupled to a reservoir in a quantum or  classical {\em thermal} state, the spin undergoes partial decoherence if $p>2$ and full decoherence if $0<p<2$ (for $p\le 0$ the thermal state is not defined). Moreover, the ratio of the decoherence functions of the spin coupled to a quantum and classical thermal reservoir is, for $t\rightarrow\infty$,
\begin{equation}
\frac{D_\varepsilon(t)}{D_0(t)} \propto
e^{-\frac1\pi\varepsilon\lambda\Gamma_\infty}, \qquad p>0.
\label{50}
\end{equation}
The decoherence speed is thus the same in the two cases.
\end{itemize}
\end{cor}

{\bf Proof of Corollary \ref{cor:2}.} Point (a) follows directly from \eqref{de1}, \eqref{52} and Corollary \ref{cor:deco}. Point (b) follows from \eqref{de1.0}, \eqref{52.0} and \eqref{49}. Finally we show (c). As remarked after \eqref{44}, the decoherence function is defined for $p>0$ only in the thermal case. We write
\begin{eqnarray}
\label{ratio}
\frac{D_\varepsilon(t)}{D_0(t)} &=& \exp\Big[ -\tfrac{\lambda^2}{4}\, \langle g_t, \Big[ \varepsilon\coth(\beta'\varepsilon\omega/2) - 2/(\beta'\omega)\big]g_t\rangle\Big]\\
&=&\exp \Big[ -\tfrac{\lambda^2}{2}\int_0^\infty S_\varepsilon(\omega)\frac{1-\cos(\omega t)}{\omega^2}d\omega\Big],
\nonumber
\end{eqnarray}
where 
$$
S_\varepsilon(\omega) = \frac{4}{\pi}\big[\varepsilon\coth(\beta'\varepsilon\omega/2) - 2/(\beta'\omega)\big] \mathcal J(\omega).
$$
For small $\omega$ we have the expansion $\varepsilon\coth(\beta'\varepsilon\omega/2) - 2/(\beta'\omega) = \tfrac16 \varepsilon^2\beta'\omega +\cdots$ and so $S_\varepsilon(\omega)\propto \omega \mathcal J(\omega)$ as $\omega\rightarrow 0$. Taking $\mathcal J(\omega)=\omega^p G(\omega)$, $p>0$, as in Corollary \ref{cor:deco} gives $S_\varepsilon(\omega) \propto \omega^{p+1}G(\omega)$. Then we apply theorem \ref{thm:trush} to arrive at \eqref{50}. This completes the proof of Corollary \ref{cor:2}.\hfill $\blacksquare$

\subsection{Examples and Illustrations.}

\begin{itemize}
\item[(1)] Consider a coherent state concentrated on a single one-particle wave function $f_0$ ($\mu_0$ is the Dirac measure concentrated on $f_0$). Then the decoherence function \eqref{de1.0} satisfies $D_0(t)=1$ for all times and all $\lambda$. As shown in Corollary \ref{cor:deco} above, $e^{-\frac14\varepsilon\lambda^2\|g_t\|^2}$ converges to zero for generic form factors $g$, so that the decoherence function \eqref{de1} satisfies $\lim_{t\rightarrow \infty} D_\varepsilon(t)= 0$ for all $\varepsilon>0$. The quasi-classically coupled spin does not decohere at all, while any degree of quantumness $\varepsilon>0$ produces full decoherence. 

\item[(2)] 
Take the initial state \eqref{m13.1} for a mixture of coherent states $\Psi_{e^{-i\theta}f_0}$ where $f_0\in\h$ is fixed and the $\theta\in[-\pi,\pi]$ is chosen uniformly in this interval. As remarked in Section \ref{sect:BEC}, in the classical limit, this is the same case as the Bose-Einstein condensate. From \eqref{de1} we have 
\begin{equation}
\label{m29}
D_0(t) = \int_{-\pi}^\pi \frac{d\theta}{2\pi} \ e^{i\sqrt 2\lambda\, {\rm Re}\, e^{\i\theta} \langle f_0,g_t\rangle} 
= J_0\big(|q(t)|\big),\qquad q(t) = \sqrt2\lambda \langle f_0,g_t\rangle.
\end{equation}
We have used the  integral representation of the Bessel function (see \cite{MF} (5.3.66)), for $z=x+iy\in\cx$,
\begin{equation}
\label{m31}
J_0(|z|) = \int_{-\pi}^\pi\frac{d\theta}{2\pi} e^{i \, {\rm Im} \, z e^{i\theta}} =\int_{-\pi}^\pi\frac{d\theta}{2\pi} e^{i \,(x\sin\theta +y\cos\theta) }.
\end{equation}
\begin{figure}[h]
\centering
\includegraphics[width=10cm]{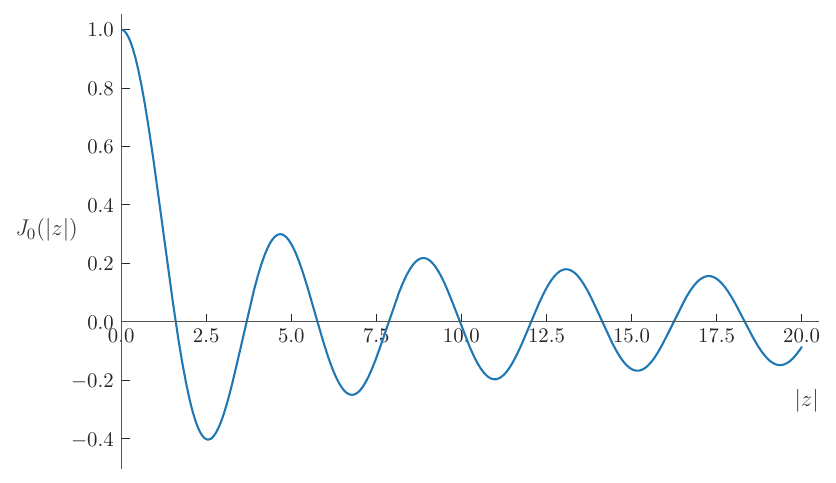}
\captionsetup{labelformat=empty}
\caption{Fig.1:\ Graph of the Bessel function of the first kind $J_0(|z|)$.}
\label{Figure 1}
\end{figure}
The value of the integral in \eqref{m31} depends on the modulus $|z|$ only, as for any $\alpha\in\rx$,
$$
\int_{-\pi}^\pi\frac{d\theta}{2\pi} e^{i \, {\rm Im} \, z e^{i\alpha} e^{i\theta}} = \int_{-\pi}^\pi\frac{d\theta}{2\pi} e^{i \, {\rm Im} \, z e^{i(\theta+\alpha)}}= \int_{-\pi}^\pi\frac{d\theta}{2\pi} e^{i \, {\rm Im} \, z e^{i\theta}},
$$
where the last equality holds since the value of the integral is the same when integrated over any interval of size $2\pi$, the period of the integrand as a function of $\theta$. In particular, we can replace $z$ by $iz$ in the integral representation, which is the same as replacing $\rm Im$ with $\rm Re$ in the exponent of the integrand in \eqref{m31}. 

As $J_0(|q(t)|)<1$ for $t>0$, the system does experience a loss of coherence, but it does generally not fully decohere.  Indeed, if $f_0(k)$ is a function such that $\langle f_0, \tfrac1\omega g\rangle$ is finite, then by the Riemann-Lebesgue Lemma\footnote{
In the case $k\in\rx^3$, $\omega(k)=|k|$ and $J(0)>0$, the form factor behaves as $|g(k)|\sim |k|^{-1}$ for small $k$. The condition \eqref{RLL} then implies that for small $\alpha$, $f_0(k)\sim |k|^\alpha$ for some $\alpha>-1$. For a general square integrable function $f(k)$ the infrared behaviour is $f(k)\sim |k|^{\alpha'}$, some $\alpha'>-3/2$. The condition $f_0\sim |k|^\alpha$ to have partial decoherence only (namely, $\alpha>-1$) is  stronger, meaning that the slow reservoir modes should not be too strongly populated in the initial reservoir coherent state built from $f_0$. This finding is in agreement with previous results obtained in related models, albeit for quantum reservoirs ($\varepsilon =1$ in our setting) \cite{PSE, MBS+, Trush}, where it was observed that the decoherence is only partial provided the  presence of (or the coupling to) the infrared reservoir modes is not strong enough. }, 
\begin{equation}
\label{RLL}
\lim_{t\rightarrow\infty} q(t) =\sqrt 2\lambda \langle f_0, \tfrac1\omega g\rangle,
\end{equation}
and so the asymptotic value of the decoherence function is
\begin{equation}
\label{62}
\lim_{t\rightarrow\infty} D_0(t) = J_0\big(\sqrt2 |\lambda|\, |\langle f_0,\tfrac{1}{\omega}g\rangle| \big),
\end{equation}
which is generically nonzero, except when $\sqrt2 |\lambda|\, |\langle f_0,\tfrac{1}{\omega}g\rangle|$ happens to be one of the zeros of $J_0$. Those zeroes form a discrete set (Figure \ref{Figure 1}).  Moreover, as 
$$
J_0(z)\sim \frac{\cos(|z|-\pi/4)}{\sqrt{|z|}},\qquad |z|\rightarrow\infty,
$$ 
we see that asymptotically in time the decoherence function decays as $|\lambda|^{-1/2}$ for large coupling constants $\lambda$. This means that the stronger the coupling to the reservoir, the bigger the loss of coherence of the spin, which is naturally expected. If $\lambda$ is small, then \eqref{62} is $1+O(\lambda)$, which expresses  the stability of the free dynamics. On the other hand, the exponential factor in the quantum case \eqref{52} converges to zero for generic form factors (see Corollary \ref{cor:deco}) while $L_n(\frac12\varepsilon\lambda^2|\langle f_0,g_t\rangle|^2)$ is bounded. Thus $D_\varepsilon(t)$ converges to zero as $t\rightarrow\infty$. We conclude that the quasi-classical spin undergoes partial decoherence while the quantum spin has full decoherence. 

\end{itemize}

\subsection{Non-Markovianity}
\label{sec:nm}

To analyze the (non-)Markovian behaviour of the spin dynamics we use the measure for Markovianity introduced in \cite{Laine-et-al}. The basic idea behind this measure is the following. The
distinguishability between two quantum states $\rho$ and $\nu$ is measured by the trace distance $\|\rho-\nu\|_1$, where $\|A\|_1={\rm tr}\sqrt{A^*A}$.  The smaller the value of $\|\rho-\nu\|_1$ the less the states $\rho$, $\nu$ are distinguishable. For a Markovian dynamics the distinguishability between two quantum states never decreases during the evolution.  In \cite{Laine-et-al} the following measure of non-Markovianity is introduced,
\begin{equation}
\mathcal N = \max_{\rho(0),\nu(0)}\int_{S_+} \partial_t \|\rho(t)-\nu(t)\|_1 dt,
\end{equation}
where the maximum is taken over all initial density matrices $\rho(0), \nu(0)$ and
\begin{equation}
\label{splus}
S_+=\{t\ge 0\ :\ \partial_t \|\rho(t)-\nu(t)\|_1>0\}.
\end{equation}
Whenever $\mathcal N>0$ the dynamics is called non-Markovian. In the current work, we do not attempt to find the actual value of $\mathcal N$. Rather, we focus on the region $S_+$ of times $t$ during which the process is non-Markovian. $S_+$, defined in \eqref{splus}, is actually {\em independent} of the initial values $\rho(0)$, $\nu(0)$, it depends only on the decoherence function, as we show below. 
\medskip

Let $\gamma, \delta$ be two $2\times 2$ density matrices with matrix elements $[\cdot]_{ij}$. Then
$$
\big\| \gamma-\delta \big\|_1 = \tr |\gamma-\delta| =\sqrt{(\gamma_{11}-\delta_{11})^2 +|\gamma_{12}-\delta_{12}|^2}.
$$
As the diagonal density matrix elements are time-independent in our energy-conserving model, we obtain
$$
\partial_t \big\| \gamma(t)-\delta(t) \big\|_1 = \frac12 \frac{\partial_t|v(t)|^2}{\sqrt{d^2+|v(t)|^2}},
$$
where 
$d=\gamma_{11}-\delta_{11}$, $v=\gamma_{12}-\delta_{12}$. Then 
$$
\partial_t\big\| \gamma(t)-\delta(t) \big\|_1 >0 \ 
\Longleftrightarrow\  \partial_t |v(t)|^2>0.
$$
From Proposition \ref{prop1} we have for $\varepsilon>0$, $|v(t)|=|D_\varepsilon(t)|\, |v(0)|$, so that for $v(0)\neq 0$
\begin{equation}
\label{m34}
\partial_t\big\| \gamma(t)-\delta(t) \big\|_1 >0 \ \Longleftrightarrow \ 
 \partial_t |D_\varepsilon(t)|^2>0 \ \Longleftrightarrow \ 
 \partial_t |D_\varepsilon(t)|>0.
\end{equation}
This shows that $S_+$, \eqref{splus}, is equivalently given by 
\begin{equation}
\label{splus1}
S_+ = \{t\ge 0\ :\ \partial_t |D_\varepsilon(t)| >0\}.
\end{equation}
From \eqref{splus}, \eqref{splus1} we conclude that the moments in time where non-Markovianity is built up are exactly the moments in time where coherence is increased. We now analyze the sets $S_+$ in more detail for the Bose-Einstein condensate state and the thermal state, as a function of $0\le \varepsilon\le 1$.

\subsubsection{Bose-Einstein condensate}  The time derivative of the decoherence function \eqref{52} is given by
\begin{align*}
\partial_t (|D_{\varepsilon} (t)|^2 ) & =  \\  = e^{- \frac{1}{2} \varepsilon || g_t ||^2} & \Biggl( -\frac{1}{2} \varepsilon | \langle f_0 | g_t \rangle |^2 {\rm Re} \big(\langle e^{i t \omega} g | g_t \rangle\big) L_n\Big(\frac{1}{2} \varepsilon | \langle f_0 | g_t \rangle |^2 \Big)^2 + \\ 
+ \varepsilon \, & {\rm Re} \big( \langle f_0 | g_t \rangle \langle f_0 | e^{i t \omega} g \rangle \big) L_{n} \Big(\frac{1}{2} \varepsilon | \langle f_0 | g_t \rangle |^2 \Big) L_{n} ' \Big(\frac{1}{2} \varepsilon | \langle f_0 | g_t \rangle |^2 \Big) \Biggr),
\end{align*}
where $ \varepsilon=\frac{1}{n} $ and we have taken $\lambda=1$. 
For concreteness of the numerical simulations, we take the form factor $g$ and the single particle state $f_0$ to satisfy
\begin{equation}
 g(\omega) = f_0(\omega) = \frac{1}{\omega_c^{(\alpha +3)/ 2} (\Gamma(\alpha+3) )^{\frac{1}{2}}} e^{-\frac{\omega}{2 \omega_c}} \omega^{\frac{\alpha}{2}},
\label{gf0}
\end{equation}
which is a radial function in $ L^2( \rx^3 ; dk )$ with $ \omega = |k|$ and where $\Gamma(\cdot)$ is the `gamma function' (reducing to the factorial function for integer arguments). The  parameter $ \alpha >-3$  characterizes the infra-red behaviour, $g(\omega)\sim \omega^{\alpha/2}$ for $\omega$ small, which amounts to $\mathcal J(\omega)\sim\omega^{2+\alpha}$ for the spectral density \eqref{specden}. The parameter $ \omega_c > 0$ is a smooth  ultra-violet cutoff. The prefactor in \eqref{gf0} is chosen so that the function is normalized in $ L^2( \rx_+, \omega^2 d \omega )$.
\medskip

In Fig.~\ref{fig:non-markBEC} we plot $\partial_t|D_\varepsilon(t)|^2$ for varying $\varepsilon=1/n$ and times $t$. We see that 
\begin{itemize}
\item[$\bullet$] Starting with an initial phase of Markovianity the spin dynamics oscillates between being Markovian and non-Markovian, with decreasing amplitude. 

\item[$\bullet$] The frequency of the oscillation between the two regimes is faster if the reservoir infra-red modes are weakly coupled to the spin: In panels (a) and (b), the infra-red modes of the reservoir are strongly coupled to the spin (infra-red singular coupling, $\alpha = -2.9$ in \eqref{gf0}), while in panels (c) and (d) they are weakly coupled ($\alpha=40$).

\item[$\bullet$] As $\varepsilon$ decreases from $\varepsilon =1$ (quantum) to $\varepsilon =0$ (classical) the regions of time where non-Markovianity is accumulated quickly stabilizes to become $\varepsilon$-independent. We do not detect an obvious universal ($\alpha, \omega_c$ independent) relation between the size of $\varepsilon$ and time regions of non-Markovianity. This last point is further illustrated in Fig.~\ref{Fig:non-Mark-crossBEC}, where we compare $\partial_t|D_\varepsilon(t)|^2$ for various values of $\varepsilon$. The graphs superpose cross-sections, slices of the graphs in Fig.~\ref{fig:non-markBEC} for a few values of $\varepsilon$. 
\end{itemize}

\begin{figure}[h]
(a)
\begin{subfigure}{0.48\textwidth}
\centering
\includegraphics[width = \textwidth]{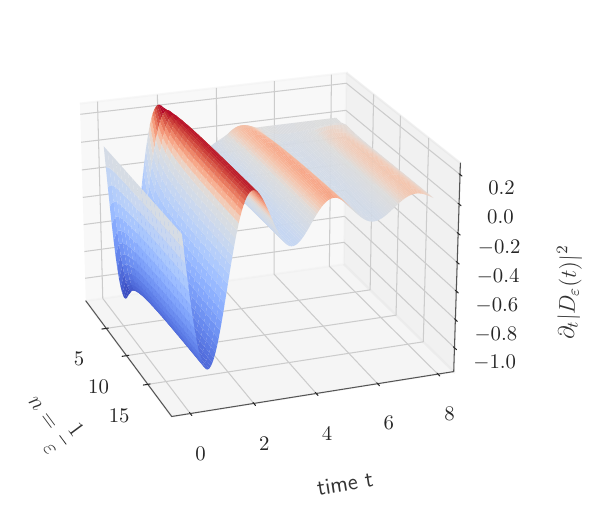}
\end{subfigure}
(b)\begin{subfigure}{0.42\textwidth}
\centering
\includegraphics[width = \textwidth]{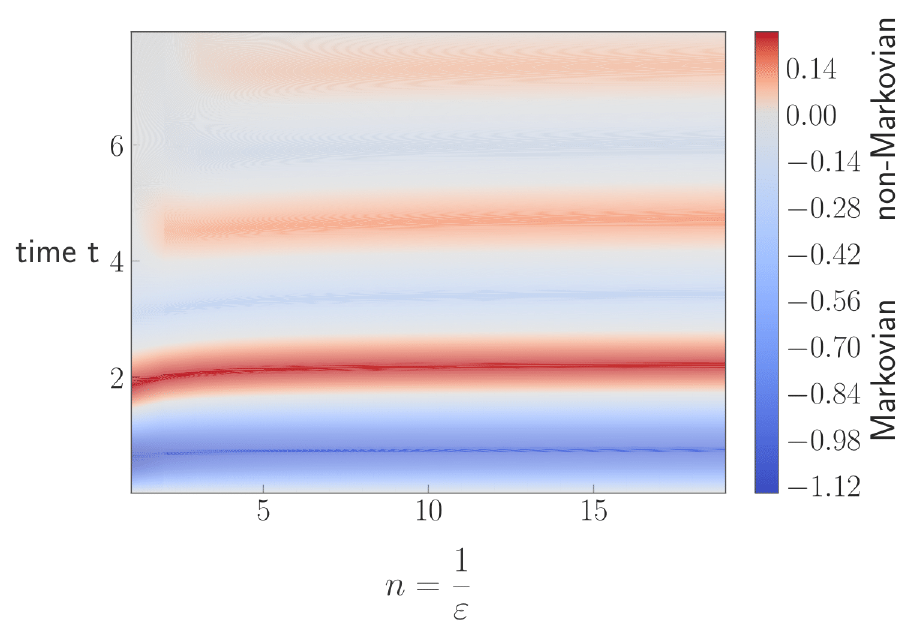}
\end{subfigure}
(c)\begin{subfigure}{0.48\textwidth}
\centering
\includegraphics[width = \textwidth]{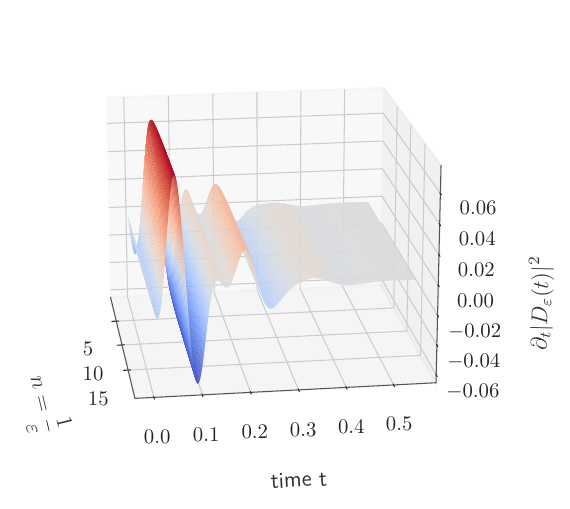}
\end{subfigure}
(d)\begin{subfigure}{0.42\textwidth}
\centering
\includegraphics[width = \textwidth]{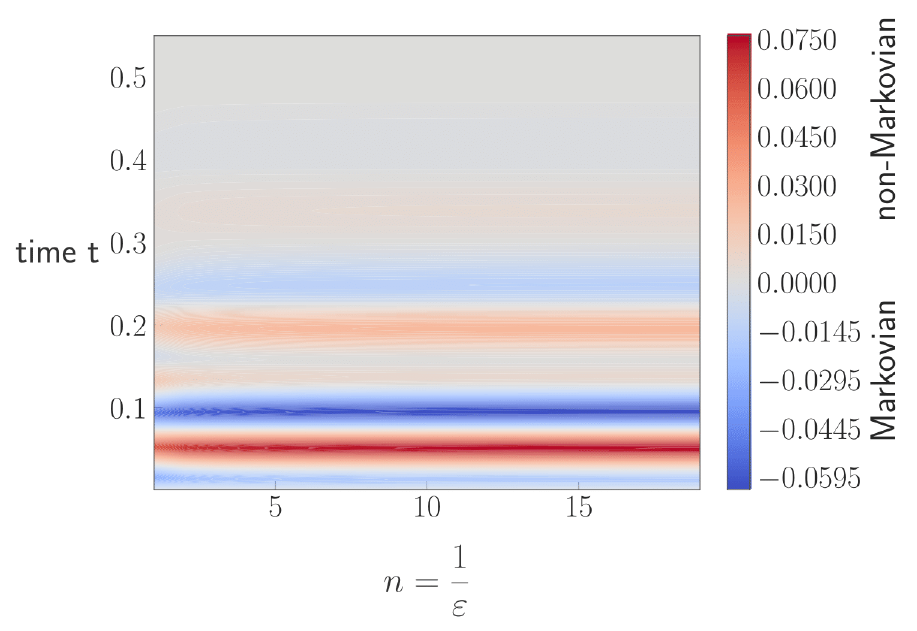}
\end{subfigure}
\captionsetup{labelformat=empty}
\caption{\label{fig:non-markBEC}Fig.2: Regions of non-Markovianity of the spin coupled to a BEC reservoir, as a function of time $t$ and classical parameter $\varepsilon$.  The BEC wave function $f_0$ and the form factor $g$ are given in \eqref{gf0}, with $\omega_c=1$.  In (a) and (b) we take $\alpha=-2.9$. In (c) and (d) we take  $\alpha=40$. Times for which $\partial_t|D_\varepsilon(t)|^2>0$ (red colour in the online version) are times in which non-Markovianity is built up.}
\end{figure}

\begin{figure}[h]
\centering
\begin{subfigure}{0.49\textwidth}
\centering
\includegraphics[width = \textwidth]{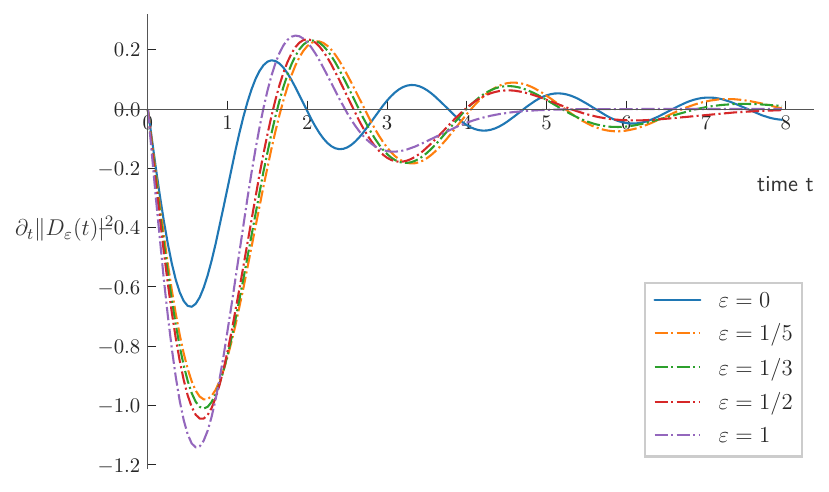}
\end{subfigure}
\begin{subfigure}{0.49\textwidth}
\centering
\includegraphics[width = \textwidth]{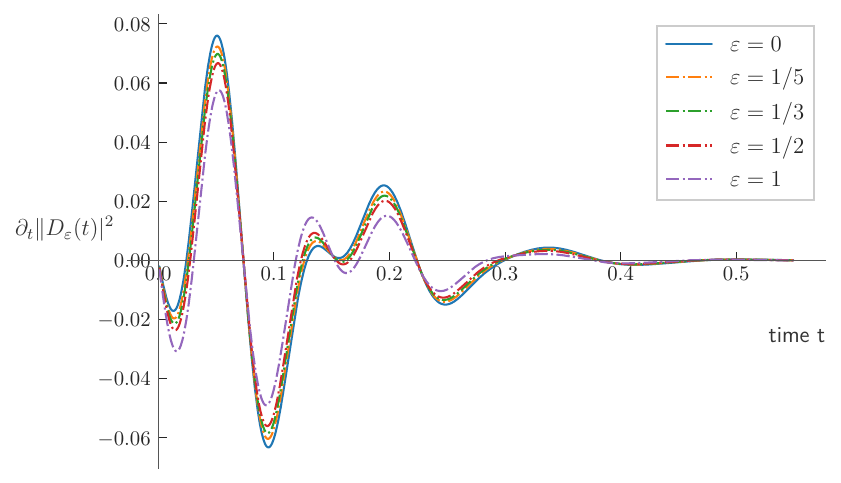}
\end{subfigure}
\captionsetup{labelformat=empty}
\caption{\label{Fig:non-Mark-crossBEC}Fig.~3: 
Comparison of regions of non-Markovianity for the dynamics of the spin coupled to the BEC, for various values of the parameter $\varepsilon$ between $0$ and $1$. $\varepsilon=0$ is the classical limit. The BEC wave function $f_0$ and the form factor $g$ are given in \eqref{gf0}. Left panel: $\alpha=-2.9$ and $\omega_c=1$. Right panel: $\alpha=40$ and $\omega_c=0.2$.}
\end{figure}

\subsubsection{Thermal state} Starting from \eqref{de3},
\begin{eqnarray}
    \partial_t | D_\varepsilon (t) |^2 &=& -\frac{2\varepsilon \lambda^2}{\pi} 
    \int_0^\infty \frac{\sin(\omega t)}{\omega}\coth\Big(\frac{\beta'\varepsilon\omega}{2}\Big)\mathcal J(\omega) d\omega\nonumber\\
    && \times \exp\Bigg[ -\frac{2\varepsilon \lambda^2}{\pi} \int_0^\infty \frac{1-\cos(\omega t)}{\omega^2}\coth\Big(\frac{\beta' \varepsilon\omega}{2}\Big)\mathcal J(\omega) d\omega\Bigg],
    \label{63}
\end{eqnarray}
where the spectral density $\mathcal J(\omega)$ is given in \eqref{specden}. As $\mathcal J(\omega)\ge 0$ it is clear that for small values of $t\ge 0$, the dynamics is Markovian, since $\partial_t | D_\varepsilon (t) |^2 \le 0$. The duration of this initial period of Markovianity depends on the infra-red behaviour of the spectral density: If infra-red modes are strongly coupled, meaning that $\mathcal J(\omega)$ is large for small $\omega\sim 0$ (divergent as $\omega\rightarrow 0$), then the first integral in \eqref{63} is large, positive for small values of $t$. If the infra-red modes are suppressed in the coupling ($\mathcal J(\omega)$ small for small $\omega$) then the positive value of that integral is smaller. This indicates that the stronger the infra-red modes are coupled to the spin the longer the duration of the initial time window of Markovianity. The conclusion is valid for all values of $0\le \varepsilon \le 1$. 

Our numerical calculations given in Fig.~\ref{fig:non-mark} illustrate this point. They show the following:
\begin{itemize}
\item[$\bullet$] For strongly coupled infra-red modes (panels (a), (b)) the dynamics is Markovian for all times. As the infra-red modes become less strongly coupled, some non-Markovianity is built up in time (panels (c), (d)). 

\item[$\bullet$] The build up of non-Markovianity (for weakly coupled infra-red modes, panels (c), (d)) depends on $\varepsilon$. For values of $\varepsilon$ close to $1$ -- the quantum case -- the non-Markovianity is enhanced.  
\end{itemize}

\begin{figure}[h]
(a)\begin{subfigure}{0.48\textwidth}
\centering
\includegraphics[width = \textwidth]{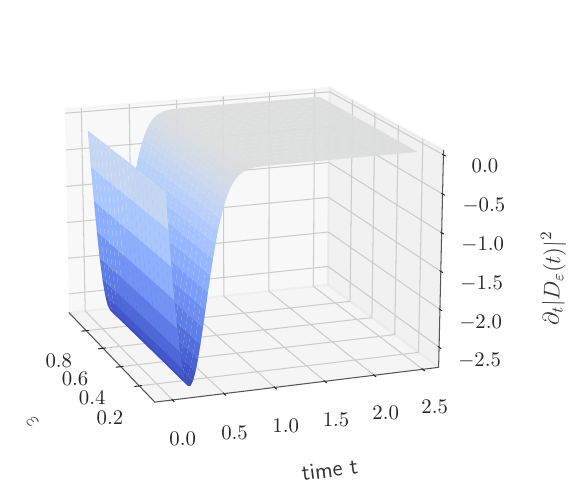}
\end{subfigure}
(b)\begin{subfigure}{0.42\textwidth}
\centering
\includegraphics[width = \textwidth]{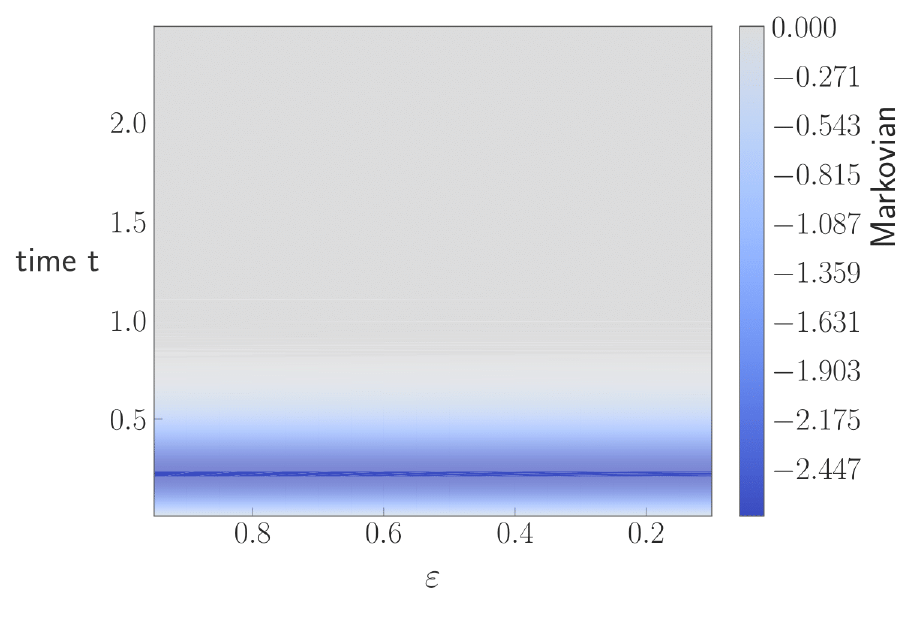}
\end{subfigure}
(c)\begin{subfigure}{0.48\textwidth}
\centering
\includegraphics[width = \textwidth]{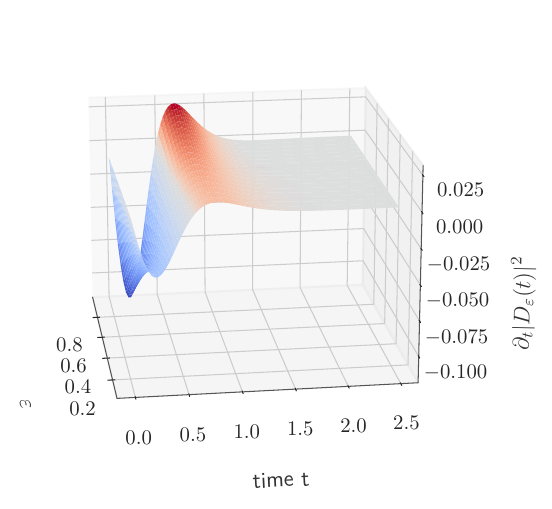}
\end{subfigure} 
(d)  \begin{subfigure}{0.42\textwidth}
\centering
\includegraphics[width = \textwidth]{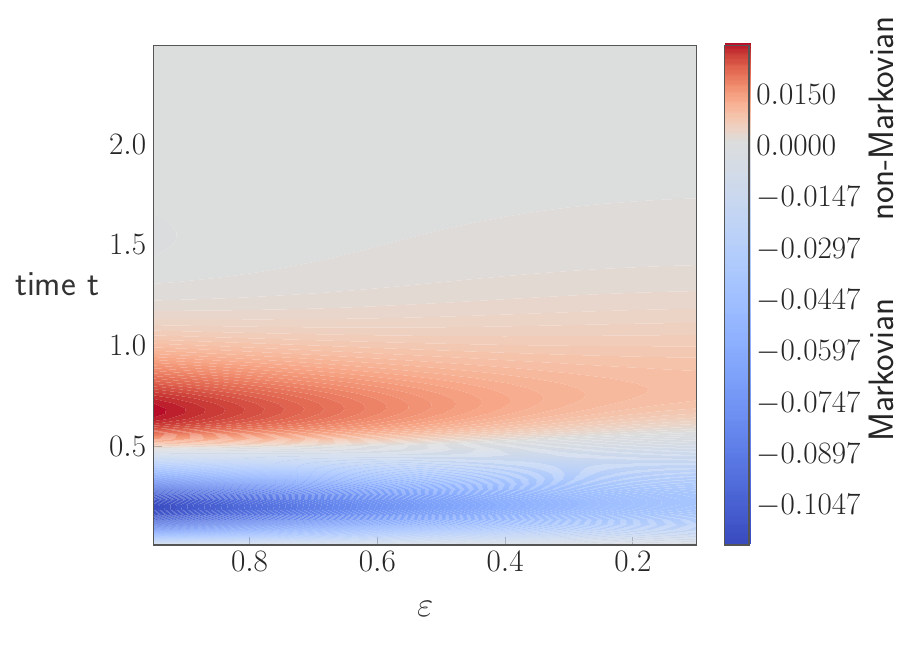}
\end{subfigure}
\captionsetup{labelformat=empty}
\caption{\label{fig:non-mark}Fig.~4: Regions of non-Markovianity of the spin dynamics coupled to the thermal reservoir, as a function of time $t$ and classical parameter $\varepsilon$. The decoherence function $D_\varepsilon(t)$ is given in \eqref{de3}, in which we take $\lambda=\beta'=1$.  The spectral density $\mathcal J(\omega)$ is given by with $g$ as in \eqref{gf0}. For panels (a) and (b) the infra-red behaviour is $\alpha = -1.9$ while for (c) and (d) it is $\alpha=5$.}
\end{figure}

\subsection{Benchmarking}
\label{sec:bench}

The goal of this section is to check directly that the dynamics of the spin obtained by the general quasi-classical theory outlined in Section \ref{sect:qcCFO} is the same as that obtained by taking the quasi-classical limit of the exact dynamics given in Proposition \ref{prop1}.

\begin{prop}
\label{prop:bench}
Consider the energy-conserving model \eqref{m1} with the classical reservoir state determined by the characteristic functional $\chi_0 = \lim_{\varepsilon\rightarrow 0}\chi_\varepsilon$ (as in Section \ref{secCF}). Then the spin dynamics given by \eqref{m18} coincides with the spin dynamics  obtained from the explicit solution given in Proposition \ref{prop1} in the limit $\varepsilon\rightarrow 0$. 
\end{prop}

{\em Proof of Proposition \ref{prop:bench}.\ }
For the energy-conserving spin-Boson model \eqref{m1}, the unitary  $U_t:=U_{t,0}$ (see \eqref{m21.1}, \eqref{m21.2}) is the solution of 
\begin{equation}
\label{n1}
i\partial_t U_t(f) = \big[\tfrac12\omega_0\sigma_z +\tfrac{\lambda}{\sqrt 2} \alpha_t(f) \sigma_z \big]U_t(f),\qquad U_0(f)=\bbbone,
\end{equation}
with $\alpha_t(f)={\rm Re}\langle e^{-it\omega} f,g\rangle$. Recall that $f\in\h$ is to be integrated over ({\em c.f.} \eqref{m18}) and $g$ is the form factor (see \eqref{m1}). We can solve \eqref{n1} explicitly. Writing $\sigma_z|1\rangle=|1\rangle$ and $\sigma_z|2\rangle = -|2\rangle$ we have 
$$
U_t(f) = e^{-\frac{i}{2} t\omega_0} e^{-\frac{i\lambda}{\sqrt 2} \int_0^t \alpha_s(f) ds} |1\rangle\langle1| + e^{\frac{i}{2} t\omega_0} e^{\frac{i\lambda}{\sqrt 2} \int_0^t \alpha_s(f) ds} |2\rangle\langle 2|.
$$
Carrying out the integrals gives
\begin{eqnarray}
U_t(f) = e^{-\frac{i}{2} t\omega_0} e^{\frac{i\lambda}{\sqrt 2} {\rm Re}\langle f, \frac{1-e^{i\omega t}}{i\omega} g\rangle } |1\rangle\langle 1|  + e^{\frac{i}{2} t\omega_0} e^{-\frac{i\lambda}{\sqrt 2} {\rm Re}\langle f, \frac{1-e^{i\omega t}}{i\omega} g\rangle } |2\rangle\langle 2|.
\end{eqnarray}
Thus,
\begin{eqnarray}
\label{m24.1}
[U_t(f) \gamma U_t(f)^*]_{12} = e^{-i t \omega_0}e^{i\sqrt 2\lambda{\rm Re} \langle f, \frac{1-e^{i\omega t}}{i\omega} g\rangle}\gamma_{12}.
\end{eqnarray}
It follows that the decoherence  according to the formula \eqref{m18} is given by  
\begin{equation}
\label{dynqc2}
[\gamma(t)]_{12} = e^{- i t \omega_0}\widetilde D_0(t)[\gamma(0)]_{12},
\end{equation}
with 
\begin{equation}
\label{dqc2}
\widetilde D_0(t)= \int_\h d\mu(f)\,  e^{i\sqrt 2\lambda{\rm Re} \langle f, \frac{1-e^{i\omega t}}{i\omega} g\rangle} = \chi_0\big(\lambda \tfrac{1-e^{i\omega t}}{i\omega}g\big),
\end{equation}
where the last equality is due to \eqref{16}. It follows that the two spin dynamics in question are equal in the limit $\varepsilon\rightarrow 0$ by comparing \eqref{dqc2} with \eqref{m12}. \hfill $\blacksquare$

\section{The energy exchange model}
\label{sec:exmod}

While the energy conserving model discussed in Section \ref{sec:econmodel} is explicitly solvable, this is not the case for the energy exchange model
\begin{equation}
\label{m27}
H(\varepsilon)= \tfrac{\omega_0}{2} \sigma_z+\d\Gamma(\omega) + \lambda  G \otimes \varphi_\varepsilon(g),
\end{equation}
where $G$ is an arbitrary matrix acting on the spin, not necessarily commuting with $\sigma_z$. A typical example is $G=\sigma_x$, the Pauli $x$-operator.  The quasi-classical theory of Section \ref{sect:qcCFO} applies and gives the spin evolution (c.f. \eqref{m18} - \eqref{m21.2})
\begin{equation}
	\label{eq: z unitary}
i\partial_t U_{t,s}(f) = \left[ \tfrac12 \omega_0 \sigma_z+\sqrt 2 \lambda \alpha_t(f) G \right] U_{t,s}(f), \qquad U_{t,t}(f)=\bbbone
\end{equation}
where 
\begin{equation}
\label{alpha}
\alpha_t(f)={\rm Re}\langle e^{-it\omega} f,g\rangle.
\end{equation}
If the functions $\omega(k)$, $f(k)$ and $g(k)$ are such that $\alpha_t(f)\rightarrow 0$ in the limit $t\rightarrow\infty$, then the `interaction term' $\propto \lambda$ in \eqref{eq: z unitary} becomes small for large times. Physically, the property $\alpha_t(f)\rightarrow 0$ for large times can be expected because $\alpha_t(f)$ is the (real part of the) probability amplitude between the freely evolved function $e^{-it\omega}f$ and the form factor $g$ and if the free field dynamics is `dispersive' then such amplitudes converge to zero in the large (positive or negative) time limit, similarly to what happens in scattering theory. Mathematically, $\alpha_t(f)\rightarrow 0$ as $t\rightarrow\infty$ is justified by the Riemann-Lebesgue lemma.

\subsection{Scattering regime}

Recall that $H_0$ is the uncoupled Hamiltonian, \eqref{Hamilt}. Denote the free and interacting propagators by 
$$
U_0(t,s) = e^{-i (t-s)H_0},\quad U(t,s) \equiv U_{t,s}(f),
$$
where $U_{t,s}(f)$ is the solution of \eqref{eq: z unitary} for a fixed $f$. The wave operators $\Omega_\pm$ are defined by the limits (if they exist) 
\begin{equation}
\label{wops}
\Omega_+ = \lim_{t\rightarrow\infty} [U(t,0)]^* U_0(t,0),\qquad \Omega_- = \lim_{s\rightarrow -\infty} [U(0,s)]^*U_0(0,s).
\end{equation} 
The motivation for the definition is as follows \cite{Dollard, Kato, H1}. In a scattering process we expect that given any state $\Psi_0$ there exist an $f_-$ and an $f_+$ such that 
\begin{equation}
\label{scatt1}
\lim_{s\rightarrow-\infty}\| U(0,s)\Psi_0 - U_0(0,s)f_-\| =0\quad \mbox{and}\quad \lim_{t\rightarrow\infty}\| U(t,0)\Psi_0 - U_0(t,0)f_+\| =0.
\end{equation}
Here, $\Psi_0$ represents the system at time zero (around which the scattering is happening) while $f_-$ is an asymptotically freely evolving state (called the `in-state') and $f_+$ is another asymptotically freely evolving (so-called `out'-) state. The relations \eqref{scatt1} mean that 
$$
\Psi_0=\lim_{s\rightarrow-\infty} [U(0,s)]^*U_0(0,s)f_- = \lim_{t\rightarrow\infty} [U(t,0)]^*U_0(t,0)f_+
$$
from which we get
$$
\Omega_-f_- = \Omega_+f_+,\qquad \mbox{or}\qquad f_+=Sf_-,\qquad S:=[\Omega_+]^{-1} \Omega_-.
$$
The scattering operator $S$ links the incoming to the outgoing states. Recall the definition $\alpha_t(f)$, \eqref{alpha}, and define
\begin{equation}
\label{int-decay}
\|\alpha(f)\|_1\equiv \int_0^\infty |\alpha_t(f)|dt.
\end{equation}

\begin{prop}[Existence of scattering and wave operators]
\label{prop:scatt}
Suppose that $\|\alpha(f)\|_1<\infty$. Then the wave operators $\:\Omega_+$, $\Omega_-$ exist (as limits in the operator norm sense). Moreover, if $|\lambda|\, \|\alpha(f)\|_1<1$, then $\Omega_\pm$ are invertible and the scattering operator $S$ exists. 
\end{prop}

{\bf Proof.} Integrating the equation 
$$
\partial_t [U(t,0)]^* U_0(t,0)  = i\lambda \alpha_t(f) [U(t,0)]^* G U_0(t,0)
$$
we obtain
\begin{equation}
\label{scatt2}
[U(t,0)]^*U_0(t,0) = \bbbone +i\lambda\int_0^t d\tau \: \alpha_\tau(f)\: [U(\tau,0)]^* G U_0(\tau,0)
\end{equation}
and so ($t'\le t$)
\begin{eqnarray*}
\Big\|[U(t,0)]^*U_0(t,0)-[U(t',0)]^*U_0(t',0) \Big\| &=& \Big\| \lambda\int_{t'}^t d \tau \: \alpha_\tau(f)\: [U(\tau,0)]^* G U_0(\tau,0)\Big\| \nonumber\\
&\le& |\lambda|\ \|G\|\int_{t'}^t d\tau\, |\alpha_\tau(f)|\rightarrow 0.
\end{eqnarray*}
This means that $[U(t,0)]^* G U_0(t,0)$ is Cauchy and thus has a limit as a bounded operator on $\cx^2$ as $t\rightarrow\infty$, showing that $\Omega_+$ exists. An analogous argument gives the existence of $\Omega_-$. To show that $\Omega_\pm$ are invertible, we note that the equation \eqref{scatt2} gives
$$
\Omega_+= \bbbone +i\lambda\int_0^\infty d\tau \: \alpha_\tau(f)\: [U(\tau,0)]^* G U_0(\tau,0),
$$
which is an invertible operator provided $|\lambda|\, \|\int_0^\infty d\tau\: \alpha_\tau(f) [U(\tau,0)]^* G U_0(\tau,0)\|<1$. $\Omega_-$ is shown to be invertible analogously.
\hfill $\blacksquare$ 
\medskip

When $\Omega_+$ exists and is invertible, then we have for large positive $t$, $[U(t,0)]^* U_0(t,0)\sim \Omega_+$ or equivalently, $U(t,0)\sim U_0(t,0)\Omega_+^{-1}$, where $A\sim B$ means that $\|A-B\|\rightarrow 0$. An analogous formula holds for large negative times. This means that the dynamics of the spin is asymptotically free ($t\rightarrow\pm\infty$). Furthermore, as $\Omega_\pm=\bbbone +O(\lambda)$, the interacting dynamics of the spin deviates from the free dynamics by a term of $O(\lambda)$ at most, for all times.  

\begin{prop}[Stability of the free dynamics]
\label{prop:asfree}
Let $\mu$ be the measure determining the classical state of the field, \eqref{16}. Suppose that there is a constant $c$ such that for all $f$ in the support of the measure $\mu$ we have $\|\alpha(f)\|_1<c$. Then the quasi-classical dynamics of the spin with initial state $\gamma(0)$, given by \eqref{m18} satisfies
\begin{equation}
\label{scatt3}
\sup_{t\ge 0} \| \gamma_0(t) - e^{-i t H_\s}\gamma(0) e^{i t H_\s}\|\le C|\lambda|
\end{equation}
for some constant $C<\infty$, and where $H_\s=\tfrac12\omega_0\sigma_z$.
\end{prop}
In \eqref{scatt3}, $\|\cdot\|$ is any norm on the spin density matrices (they are all equivalent since the system is finite dimensional). 
\medskip

{\bf Proof.} We have 
$$
\gamma_0(t) =\int_\h d\mu(f) U_{t,0}(f) \gamma(0) [U_{t,0}(f)]^*.
$$
Equation \eqref{scatt2} yields (operator norm) $\|U_{t,0}(f) -e^{-i t H_0}\|\le  |\lambda|\, \|\alpha(f)\|_1\le c|\lambda|$. Replacing $U_{t,0}$ by $e^{-i t H_0}$ and similarly for their adjoints and using the last inequality yields \eqref{scatt3}.\hfill $\blacksquare$
\bigskip

The result of Proposition \ref{prop:asfree} means that even when the system interacts with the reservoir for a long time (even $t\rightarrow\infty$), there are no effects on the system beyond the size of $O(\lambda)$, provided the integrability property 
\begin{equation}
\label{alphaint}
\|\alpha(f)\|_1<\infty
\end{equation} 
holds. In the next result we give a condition which ensures \eqref{alphaint}.

\begin{lem}[Sufficient condition for \eqref{alphaint}]
\label{lem1}
Let $\omega(k)=|k|$, $k\in\rx^3$ and write functions $f\in L^2(\rx^3,d^3k)$ as $f(\Sigma,\omega)$ in polar coordinates, $(\Sigma,\omega)\in S^2\times\rx_+$. Suppose that 
\begin{itemize}
\item[\rm (a)] For each $\Sigma\in S^2$, the function $\omega\mapsto f(\Sigma,\omega) g(\Sigma,\omega)$ is twice differentiable, and

\item[\rm (b)] The infrared and ultraviolet behaviour is as
\begin{eqnarray}
\label{compcond}
\omega\rightarrow 0&:& f(\Sigma,\omega) g(\Sigma,\omega) \sim \omega^p,\ \ \ \  \mbox{some $p>-1$},\\
\omega\rightarrow\infty &:& f(\Sigma,\omega) g(\Sigma,\omega) \sim \omega^{-q}, \ \ \mbox{some $q>2$}.
\end{eqnarray}
\end{itemize}
Then the bound \eqref{alphaint} holds.
\end{lem}
\bigskip

{\bf Proof.} We have $\langle e^{-i\omega t}f,g\rangle=\int_{S^2}d\Sigma\int_0^\infty d\omega\: \omega^2 e^{i\omega t} \overline{f(\Sigma,\omega)}g(\Sigma,\omega)$.  By using $e^{i\omega t} = \frac{1}{it} \partial_\omega e^{i\omega t}$ and integrating by parts in the radial variable $\omega$ twice we obtain
\begin{equation}
\label{ibp}
\langle e^{-i\omega t}f,g\rangle= (-\frac{1}{it})^2\int_{S^2}d\Sigma \int_0^\infty d\omega\: e^{i\omega t} \partial_\omega^2 \big[\omega^2\overline{f(\Sigma,\omega)}g(\Sigma,\omega)\big].
\end{equation}
Hence $\alpha_t(f)$ decays as $t^{-2}$ for large $t$, so that $\|\alpha(f)\|_1<\infty$.\hfill $\blacksquare$

\subsection{Examples and illustrations}

\begin{itemize}
\item[1.] We point out that Propositions \ref{prop:scatt} and \ref{prop:asfree} hold for the energy conserving as well as the energy exchange models.  
We have shown in Corollary \ref{cor:2} that for the energy conserving model, the spin coupled to the classical coherent state or the classical BEC undergoes partial, not full decoherence. The decoherence function is of the form $1+O(\lambda)$. This illustrates the result of Proposition \ref{prop:asfree} from a different point of view (in Corollary \ref{cor:2} we only needed condition \eqref{49} while Lemma \ref{lem1} requires the stronger assumption \eqref{compcond}). 
    
\item[2.] For a pure coherent state of the reservoir the measure $\mu$ is concentrated at a single $f_0\in\h$ (see Section \ref{sec:cs}). Then $\|\alpha(f_0)\|_1<\infty$ provided $f\equiv f_0$ and $g$ satisfy \eqref{compcond}.  We conclude that the spin dynamics is asymptotically free and the free dynamics is stable in the sense of Propositions \ref{prop:scatt} and \ref{prop:asfree}.

\item[3.] Consider the coherent state of the reservoir given by mixing the states $e^{-i\theta}f_0$ uniformly over $S^1$ (same as Bose-Einstein condensate), {\em c.f.} Sections \ref{sec:cs}, \ref{sect:BEC}. We have 
$$
|\alpha_t(e^{-i\theta}f_0)| = |{\rm Re}\, e^{-i\theta} \langle e^{-i t\omega}f,g\rangle|\le |\langle e^{-i t\omega} f_0,g\rangle|.
$$
Then $\|\alpha(e^{-i \theta}f_0)\|_1\le c$ for all $\theta\in S^1$ provided $f\equiv f_0$ and $g$ satisfy \eqref{compcond}. The spin dynamics is again asymptotically free and the free dynamics is stable (the conditions of Propositions \ref{prop:scatt} and \ref{prop:asfree} are satisfied).
\end{itemize}

\subsection{Non-scattering regime}
\label{sect:noscatt}

We have shown in Corollary \ref{cor:2} that the spin undergoes {\em full} decoherence when coupled to the classical thermal state via an energy conserving interaction. The result was obtained by explicit calculation, and it cannot be deduced from Theorem \ref{thm:CFO} for the technical reason that the measure $\mu$ in \eqref{m18} is not a true measure on $\mathfrak h$ in the thermal case. Here we have an instance where Propositions \ref{prop:scatt} and \ref{prop:asfree} do not apply and in fact the change in the spin dynamics is not limited by $O(\lambda)$.

A situation where Theorem \ref{thm:CFO} does apply, but still the free dynamics can change over time by more than $O(\lambda)$, because  \eqref{alphaint} is {\em not} satisfied, is described by {\em polaron-type models}. In its original version, the polaron model describes electrons in a crystal modeled by phonons which make up the reservoir \cite{F}. The dispersion relation is given by $\omega(k)=\omega_\r$, a positive constant; the field Hamiltonian is proportional to the number operator. In our setting, the system Hilbert space has finite dimension ($=2$ for simplicity) and so the model does not describe electrons, but rather two effective degrees of freedom interacting with phonons, making up a polaron-type model. In this setting the function (see \eqref{alpha})
$$
\alpha_t(f) =  {\rm Re} \, e^{i\omega_\r t}\langle f,g\rangle
$$
is $2\pi/\omega_\r$-periodic in time and one can analyze the time-dependent evolution equation \eqref{eq: z unitary} by means of Floquet theory. The propagator $U_t(f)$ is $2\pi/\omega_\r$-periodic in $t$ for every $f$ and consequently so is the density matrix $\gamma(t) = \int_\h d\mu(f) U_t(f)\gamma U_t(f)^*$. As a particular example we may consider initial reservoir states given by a measure $\mu(f)$ supported on $f\in\h$ such that $\langle f,g\rangle\in\rx$. Then $\alpha_t(f)=\lambda\langle f,g\rangle \cos(\omega_\r t)$ and the time-dependent evolution equation \eqref{eq: z unitary} becomes
$$
i\partial_t U_t(f) = \Big[ \tfrac12 \omega_0 \sigma_z +\tfrac{\lambda}{\sqrt 2} \langle f,g\rangle \cos(\omega_\r t)  \sigma_x \Big] U_t(f), \qquad U_0(f)=\bbbone.
$$
This is the equation of a two-level system interacting with a classical electric field $\mathcal E(t)\propto \cos(\omega_\r t)$. The solutions to this equation have been studied widely and in detail in the literature, see for instance \cite{Grifoni} and references therein. Our theory on the quasi-classical limit thus recovers some known equations for open systems in contact with classical fields. 
\medskip

We show now how another common model, a two-level system coupled to a circularly polarized classical field, can be derived using our quasi-classical approach. Consider the Hamiltonian
\begin{equation}
\label{circpol}
    H_\varepsilon = \frac{\omega_0}{2} \sigma_z + d \Gamma(\omega_R) + \tfrac{1}{2} \Big( \sigma_x \otimes \varphi_\varepsilon(g) + \sigma_y \otimes \varphi_\varepsilon (ig) \Big), 
\end{equation}
where the dispersion relation is constant, $\omega(k) = \omega_R>0$. We take the reservoir to be in a Bose-Einstein condensate state for the field, determined by $\mu = \int_{0}^{2 \pi} \frac{d \theta}{2 \pi} \delta_{e^{i \theta} f_0 }$, for a fixed wave function $f_0$ (see \eqref{76} and text thereafter). Let us take a form factor $g$ such that ${\rm Im}\langle f_0,g\rangle=0$. For $\theta$ fixed, the unitary $U_t(\theta)$ appearing in \eqref{eq: z unitary} satisfies  $i\partial_t U_t(\theta) = H(\theta,t)U_t(\theta)$ and $U_0(\theta)=\bbbone$, where the Hamiltonian is given by 
\begin{equation*}
    H(\theta,t) = \tfrac{\omega_0}{2} \sigma_z + \tfrac{1}{2}  {\rm Re}\langle f_0, g \rangle\Big( \cos (\omega_R t-\theta) \sigma_x - \sin (\omega_R t-\theta) \sigma_y \Big).
\end{equation*}
The dynamics can be solved exactly by applying the unitary transformation $ e^{-i(\omega_R t-\theta) \frac{\sigma_z}{2}}$. Namely, $V(t) = e^{-i(\omega_R t-\theta) \frac{\sigma_z}{2}} U(\theta,t)$ solves a time-independent and $\theta$-independent equation 
$$ i \partial_t V(t) = \widetilde{H} V(t) := \tfrac{1}{2} \Big( (\omega_0 + \omega_R) \sigma_z + {\rm Re}\langle f_0, g \rangle \sigma_x \Big) V(t),\qquad V(0) = e^{i\frac\theta 2 \sigma_z}.
$$
In terms of the Rabi frequency $\Omega_{\rm Rabi} = \frac{1}{2} \big( (\omega_0 + \omega_R)^2 + ({\rm Re}\langle f_0, g \rangle)^2 \big)^{\frac{1}{2}}$, we obtain
\begin{eqnarray*}
    V(t) &=&e^{-i \widetilde{H} t} e^{i\frac\theta 2 \sigma_z}\\
    &=& \cos(\Omega_{\rm Rabi} t) e^{i\frac\theta 2 \sigma_z} -i \sin (\Omega_{\rm Rabi}t) \biggl(  \frac{\omega_0 + \omega_R }{2\Omega_{\rm Rabi}} \sigma_z + \frac{{\rm Re}\langle f_0, g \rangle}{2 \Omega_{\rm Rabi}} \sigma_x \biggr) e^{i\frac\theta 2 \sigma_z}.
\end{eqnarray*}
Then $ U(\theta,t) = e^{i(\omega_R t-\theta)\frac{\sigma_z}{2}} V(t) = e^{i(\omega_R t-\theta)\frac{\sigma_z}{2}} e^{-i \widetilde{H} t} e^{i\frac\theta 2 \sigma_z} $ and the reduced system density matrix is given by (c.f. \eqref{m18})
\begin{eqnarray*}
     \gamma(t) &=& \int_{0}^{2 \pi} \frac{d \theta}{2 \pi} U(\theta,t) \gamma_0 U^*(\theta,t)   \\ 
     &=& e^{\frac i2\omega_R t\sigma_z} \biggl( \int_{0}^{2 \pi} \frac{d \theta}{2 \pi} e^{-i\frac{\theta}{2}\sigma_z} e^{-i \tilde{H} t} e^{i\frac\theta 2 \sigma_z} \gamma_0 e^{-i\frac\theta 2 \sigma_z} e^{i \tilde{H} t} e^{i\frac{\theta}{2}\sigma_z} \biggr) e^{-\frac i2\omega_R t \sigma_z}. 
\end{eqnarray*}
Carrying out the integrals for the explicit solutions for an initial density matrix 
$$
\gamma_0 = \begin{pmatrix}
  \gamma_{11} & \gamma_{12}\\
  \bar{\gamma}_{12} & 1-\gamma_{11},
\end{pmatrix}
$$
we obtain the following formulas for the components of $\gamma(t)$:
\begin{eqnarray*}
    \gamma_{11}(t) &=& \Big[\cos^2(\Omega_{\rm Rabi} t) + \frac{(\omega_0 +\omega_R)^2}{4 \Omega_{\rm Rabi}^2} \sin^2(\Omega_{\rm Rabi} t) \Big] \gamma_{11}\\
    && +\frac{({\rm Re}\langle f_0, g \rangle)^2}{4 \Omega_{\rm Rabi}^2} \sin^2 (\Omega_{\rm Rabi}t) (1-\gamma_{11}) \\
    \gamma_{12}(t) &=& e^{i \omega_R t} \Big[ \cos(\Omega_{\rm Rabi}t) -i \frac{(\omega_0 +\omega_R)}{2 \Omega_{\rm Rabi}} \sin(\Omega_{\rm Rabi}t) \Big]^2 \gamma_{12}.
\end{eqnarray*}
We see that both, the populations and the coherence, oscillate with frequency $ \Omega_{\rm Rabi}$ attaining their initial and maximal values periodically in time. The upshot of this analysis is that our quasi-classical limit of a fully quantum system plus reservoir model gives the well-known reduced equation for a two-level system interacting with a classical circularly polarized field (Hamiltonian \eqref{circpol}). Moreover, we can solve this equation exactly and we find that both the populations and the coherences are periodic in time.

\subsection{ Symmetry considerations}
\label{subs:symmetry}

According to Theorem \ref{thm:CFO}, the quasi-classical dynamics is given by
\begin{equation}
\label{m18'}
 \gamma(t) = 
\int_\h d\mu(f) U_t(f)\gamma_0 U_t(f)^*,
\end{equation}
where $\gamma_0$ is the initial density matrix of the spin, $\mu$ is the reservoir probability measure on $\h$ arising from the limit characteristic function $\chi_0$, as given in \eqref{16}, and where the unitary $U_t$ satisfies the equation \eqref{eq: z unitary} for the spin-reservoir model. We now explore some symmetry properties of $\gamma(t)$ under the hypothesis that the measure $\mu$ is even, that is,
\begin{equation}
\label{evenmeas} 
\mu(-E) = \mu(E),
\end{equation}
for any measurable set $E$ of $ \h$. 
This condition is satisfied in particular for Gaussian centered measures as well as the measure obtained through the limit $\varepsilon \rightarrow 0 $ in the Bose-Einstein condensate case, as can be seen from \eqref{76} and the related discussion.
\begin{prop}
\label{prop:sym}
Suppose that \eqref{evenmeas} holds and that $G$ is off-diagonal in the $\sigma_z$-basis. Then the diagonal and the off-diagonal density matrix elements of $\gamma(t)$ evolve independently. In particular, if $\gamma_0$ is diagonal, then so is $\gamma(t)$ for all $t$ and if $\gamma_0$ is off-diagonal, then so is $\gamma(t)$. 
\end{prop}

{\bf Proof.} For a fixed $f$, the integrand of \eqref{m18'}, $\gamma(f,t) \equiv U_t(f)\gamma_0 U_t(f)^*$, is the unique solution of 
\begin{equation}
\label{xy1}
i \partial_t \gamma(f,t) = \left[ \frac{\omega_0}{2} \sigma_z + \sqrt{2} {\rm Re} \langle e^{-i t \omega} f | g \rangle G , \gamma(f,t) \right], 
\end{equation}
with initial condition 
\begin{equation}
\label{xy2}
\gamma(f,0)=\gamma_0.
\end{equation}
The space of bounded operators on $\mathbb C^2$, denoted $\mathcal B(\mathbb C^2)$, is a Hilbert space with inner product $\langle X,Y\rangle = \tr(X^*Y)$. In the orthonormal basis $\{\tfrac{1}{\sqrt 2}\mathbf 1, \tfrac{1}{\sqrt 2}\sigma_x,\tfrac{1}{\sqrt 2}\sigma_y,\tfrac{1}{\sqrt 2}\sigma_z\}$, any element $A\in \mathcal B(\mathbb C^2)$ is decomposed as $A=A^0+A^x+A^y+A^z$, where $A^\alpha = \tfrac12\sigma_\alpha\, \tr(\sigma_\alpha A)$, $\alpha\in\{x,y,z\}$ and $A^0=\tfrac{1}{2}\tr(A)\mathbf 1$.

As the conjugation with $\sigma_z$ leaves $\mathbf 1$ and $\sigma_z$ invariant and inverts the sign of $\sigma_x$ and $\sigma_y$, we have 
\begin{equation}
\label{xy0}
\sigma_z A \sigma_z = A^0-A^x-A^y+A^z.
\end{equation}
The component $G^0\propto \mathbf 1$ does not intervene in the dynamics \eqref{xy1}, as this part commutes with $\gamma(f,t)$. We may thus take $G^0=0$, which combined with $G^z=0$ implies that $G$ is purely off diagonal in the $\sigma_z$-basis. Then by \eqref{xy0}, $\sigma_z G\sigma_z=-G$. Conjugating \eqref{xy1} with $\sigma_z$ shows that $\sigma_z\gamma(f,t)\sigma_z$ solves the  equation \eqref{xy1}  with $f$ changed to $-f$, and has the initial condition $\sigma_z\gamma_0\sigma_z$.
We decompose 
\begin{equation}
\label{decomp}
\gamma(f,t)=\gamma_d(f,t)+\gamma_{od}(f,t),
\end{equation}
where $\gamma_d(f,t)$ is the solution of \eqref{xy1} with initial condition $\gamma_0^0+\gamma_0^z$, that is the diagonal part of $\gamma_0$ and $\gamma_{od}(f,t)$ is the solution of \eqref{xy1} with initial condition $\gamma_0^x+\gamma_0^y$, that is the off-diagonal part of $\gamma_0$. Now $\gamma_d(-f,t)$ satisfies the same differential equation as $\sigma_z\gamma_d(f,t)\sigma_z$ with the same initial condition (as $\gamma_d(f,0)$ is invariant under conjugation with $\sigma_z$). By the uniqueness of the solution we have 

\begin{equation}
\label{4}
 \gamma_d(f,t) = \sigma_z\gamma_d(-f,t)\sigma_z\qquad\mbox{and}\qquad
 \gamma_{od}(f,t) = -\sigma_z\gamma_{od}(-f,t)\sigma_z,
\end{equation}
where the first negative sign in the equation involving $\gamma_{od}$ is due to the sign switch in the initial condition when conjugating with $\sigma_z$. We now take the components of \eqref{4}, and using \eqref{xy0} we arrive at
\begin{equation}
\label{xy5}
\gamma^\alpha_d(f,t) =\gamma^\alpha_d(-f,t), \quad \mbox{$\alpha=0,z$, \ while}\quad   \gamma^\alpha_d(f,t) =-\gamma^\alpha_d(-f,t),\quad\mbox{$\alpha=x,y$}
\end{equation}
and
\begin{equation}
\label{xy6}
\gamma^\alpha_{od}(f,t) =-\gamma^\alpha_{od}(-f,t), \quad \mbox{\ $\alpha=0,z$,\  while}\quad  \gamma^\alpha_{od}(f,t) = \gamma^\alpha_{od}(-f,t),\quad\mbox{$\alpha=x,y$}.
\end{equation}
Consider now a measure $\mu$ which is invariant under the transformation $f\mapsto -f$. Then we have $\int_{L^2} d\mu(f) A(f) = \int_{L^2} d\mu(f) A(-f)$ and the integral over all odd functions in $f$ vanish, so that by \eqref{xy5}, \eqref{xy6},
\begin{equation}
\label{xy7}
\int_{\h} \gamma_d^\alpha(f,t) = 0 = \int_{\h} \gamma_{od}^\beta(f,t),\qquad \mbox{for $\alpha=x,y$ and $\beta=0,z$}.
\end{equation}
By \eqref{decomp} and \eqref{xy7} we obtain 
\begin{equation}
\label{xy8}
\gamma(t)= \int_{\h} d\mu(f) \gamma(f,t) = \int_{\h} d\mu(f)\big[ \gamma_d^0(f,t) +\gamma_d^z(f,t) + \gamma_{od}^x(f,t) +\gamma_{od}^y(f,t)\big].
\end{equation}
If $\gamma_0$ is purely off-diagonal then $\gamma_d(f,t)=0$ for all $t$  and the first two integrands on the right side of \eqref{xy8} vanish. Then $\gamma(t)$ has only components along $\sigma_x$ and $\sigma_y$, so it is off-diagonal as well. Conversely, if $\gamma_0$ is diagonal, then $\gamma_{od}(f,t)=0$ for all $t$ and $\gamma(t)$ is purely diagonal, as only components along $\mathbf 1$ and $\sigma_z$ are non-vanishing in the integral \eqref{xy8}. This concludes the proof of Proposition \ref{prop:sym}.\hfill $\blacksquare$

\bigskip

{\bf Acknowledgements.} We thank two anonymous referees for taking their time to assess this work, give us valuable feedback and point out  aspects which may open up further research on the topic.  M. Correggi and M. Falconi acknowledge the supports of PNRR Italia Domani and Next
Generation EU through the ICSC National Research Centre for High Performance Computing, Big Data and
Quantum Computing. M. Correggi, M. Falconi and M. Fantechi also acknowledge the MUR grant “Dipartimento
di Eccellenza 2023-2027” of Dipartimento di Matematica, Politecnico di Milano. M. Fantechi acknowledges the support of the ``Gruppo Nazionale di Fisica Matematica (GNFM)'' section of the ``Istituto Nazionale di Alta Matematica (INdAM)'', and expresses gratitude to the Department of Mathematics and Statistics of the Memorial University of Newfoundland, where this work was completed.
M. Merkli is grateful to
the Politecnico di Milano for the hospitality and the financial support of the visit during which this work was
conceived, as well as to the Natural Sciences and Engineering Research Council of Canada (NSERC) for support through a Discovery Grant.

\bibliographystyle{plainnat}
\bibliography{referencesCoFaFaMe}

\end{document}